\documentclass[10pt,twocolumn,twoside]{IEEEtran}
 \usepackage{cancel}
\usepackage{color}
\usepackage[normalem]{ulem}
\usepackage{graphicx, bm, mathrsfs}
\usepackage{cite,array, enumerate}
\usepackage{amssymb}
\usepackage{dsfont, bm}
\usepackage{epstopdf}
\usepackage[cmex10]{amsmath}
\newtheorem{assumption}{Assumption}
\newtheorem{theorem}{Theorem}

\newtheorem{lemma}{Lemma}

\newtheorem{proposition}{Proposition}

\newtheorem{remark}{Remark}
\newtheorem{proof}{Proof}
\newtheorem{example}{Example}
\ifCLASSINFOpdf

\else

\fi

\RequirePackage[normalem]{ulem} 
\RequirePackage{color}\definecolor{RED}{rgb}{1,0,0}\definecolor{BLUE}{rgb}{0,0,1} 
\begin{document}

\title{Distributed Networked Controller Design for Large-scale Systems under Round-Robin Communication Protocol}

\author{Tao~Yu, Junlin~Xiong
\thanks{Tao. Yu and Junlin. Xiong are with the Department of Automation, University of Science and Technology of China, Hefei 230026, China. (E-mail: \tt\small yutao16@mail.ustc.edu.cn, xiong77@ustc.edu.cn, junlin.xiong@gmail.com).}
}

\maketitle

\begin{abstract}
This paper studies the distributed $\mathcal{L}_{2}$-gain control problem for  continuous-time large-scale systems under Round-Robin communication protocol. In this protocol, each sub-controller obtains its own subsystem's state information continuously, while communicating with neighbors at discrete-time instants periodically. Distributed controllers are designed such that the closed-loop system is exponentially stable and that the prescribed $\mathcal{L}_{2}$-gain is satisfied. The design condition is obtained based on a time-delay approach and given in terms of linear matrix inequalities. Finally, three numerical examples are presented to illustrate the efficiency of the proposed scheme.
\end{abstract}

\IEEEpeerreviewmaketitle

\section{Introduction}
Large-scale systems generally comprise of many interconnected subsystems. The classical centralized control strategy for large-scale systems suffers from heavy computational burdens and decentralized control techniques often offer poor system performance \cite{kazempour2013stability}. As a result, distributed control has attracted intensive attention in the past few years \cite{massioni2009distributed,sun2018distributed,boem2019distributed}. Under distributed control strategy, each sub-controller could communicate with its neighbors at time instants. Therefore, not only local state information but also information from neighbors are used to form the control input.

In most existing literatures about distributed control for large-scale systems, it is often assumed that at each time instant, each sub-controller communicates with all neighbors simultaneously \cite{zhang2016energy,millan2019distributed,van2015synthesis}. However, this assumption is generally difficult to be satisfied when the communication energy and resources are limited \cite{ugrinovskii2014Round}. Therefore, different communication protocols are needed to orchestrate the communication order among the neighbors. These protocols
include, but are not limited to Round-Robin communication protocol \cite{zou2017state}, weight try-once-discard protocol \cite{zou2017ultimate}, stochastic communication protocol \cite{donkers2012stability} and gossip communication protocol \cite{yutao}.

Among the above various communication protocols, Round-Robin communication protocol is widely used for information transmission in networked control systems \cite{ugrinovskii2014Round}.
Here, we introduce Round-Robin communication protocol into the distributed control of large-scale systems. Under such a protocol, each sub-controller uses local information continuously, while communicating with neighbors at discrete-time instants periodically. That is, each sub-controller requires only one neighbor's latest state information at each time instant. Therefore, less transmission packets are needed and network bandwidth can be saved.

Generally speaking, there have been two approaches to address the control problems of systems under Round-Robin communication protocol. One is to describe them as hybrid systems, such as the issues about input-output stability properties of networked control systems \cite{nesic2004input}, the tradeoffs between transmission intervals, delays and performance of networked control systems \cite{heemels2010networked} and distributed state estimation over sensor networks \cite{xu2018finite}. The other one is to transform them into time-delay systems, such as the controller design and $\mathcal{L}_{2}$-gain analysis of networked control systems \cite{liu2012stability}, distributed state estimation with $H_{\infty}$ consensus \cite{ugrinovskii2014Round}. However, when consider the utilization of Round-Robin protocol into the communication among sub-controllers in large-scale systems, the results in above papers can not be adopted directly to obtain the distributed controller gains.

 In order to tackle the distributed $\mathcal{L}_{2}$-gain control problem for large-scale systems under Round-Robin communication protocol, this paper further develops the time-delay techniques thanks to a skillful partition of the time interval [0, T]. Then, based on matrix manipulations and Lyapunov stability theory, sufficient conditions are established in the form of linear matrix inequalities (LMIs) such that the closed-loop system is exponentially stable with a prescribed  $\mathcal{L}_{2}$-gain. The distributed controller gains can be obtained by solving a set of LMIs. Two numerical examples show that compared with the results in \cite{ghadami2013decomposition}, our control scheme leads to 50$\%$ bandwidth savings with a slight sacrifice of the $\mathcal{L}_{2}$-gain performance under different system parameters or different number of subsystems. The last
example illustrates that our developed theory is applied to the distributed $\mathcal{L}_{2}$-gain control problem for a large-scale system with $100$ heterogeneous subsystems.

\textbf{Notation}: The set of positive integers is denoted by $\mathbb{N}^{+}$, the $n$-dimensional Euclidean space is denoted by $ \mathbb{R}^{n}$, and $\mathcal{L}_{2}[0,\infty)$ denotes the Lebesgue space of $R^{n}$-valued vector functions defined on the time interval $[0,\infty)$. The notation ${\rm He}\{Y\}$ is a matrix defined by ${\rm He}\{Y\}=Y+Y^{\top}$. We write $A>B$ $(A\geq B)$ when $A-B$ is positive definite (positive semi-definite). In symmetric block matrices, the symbol ``*'' is used to represent the symmetric terms. All matrices and vectors are assumed to have compatible dimensions if they are not explicitly specified. The symbol ${\rm mod}(a,b)$ means the remainder when $a$ is divided by $b$.

  \section{Problem Formulation}
Consider a large-scale system with $N$ subsystems, the dynamics of the $i$th $(i\in \mathbb{N}:=\{1,2 ,\ldots, N\})$ subsystem is described as
\begin{align}\label{1}
\left\{
  \begin{array}{ll}
\dot{x}_{i}(t)=A_{ii}x_{i}(t)+ \mathop{\sum}\limits_{j\in \mathbb{N}_{i}}A_{ij}x_{j}(t)+B_{i}u_{i}(t)+E_{i}w_{i}(t),\\
z_{i}(t)=C_{i}x_{i}(t)+F_{i}w_{i}(t),
  \end{array}
\right.
\end{align}
where $x_{i}(t)\in \mathbb{R}^{n_{i}}, u_{i}(t)\in \mathbb{R}^{m_{i}}, w_{i}(t)\in \mathbb{R}^{p_{i}}$ and $z_{i}(t)\in \mathbb{R}^{q_{i}}$ denote, respectively, the state vector, the control input, the external disturbance and the performance output of the $i$th subsystem. We assume that $w_{i}(t)\in \mathcal{L}_{2}[0,\infty)$. The matrices $A_{ii},A_{ij},B_{i},E_{i},C_{i},F_{i}$ in \eqref{1} are known matrices with appropriate dimensions. In addition, the symbol $\mathbb{N}_{i}:=\{i_{1}, i_{2},...,i_{d_{i}}\}\subseteq \mathbb{N} $ means the ordered neighbor set of the $i$th subsystem, and $d_{i}$ denotes the cardinality of $\mathbb{N}_{i}$.

Here, each sub-controller uses local information continuously, but the interaction with neighbors is subject to Round-Robin communication protocol. In order to illustrate Round-Robin communication protocol precisely, we define a shift permutation operator $\Pi$ on the ordered neighbor set as
\begin{align}\label{3}
\Pi\{i_{1},...,i_{d_{i}-1},i_{d_{i}}\}=\{i_{d_{i}},i_{1},...,i_{d_{i}-1}\}.
\end{align}
Furthermore, the symbol $\Pi^{k}(\mathbb{N}_{i})$ denotes the set after using $k$-times consecutive shift permutations on $\mathbb{N}_{i}$ (The superscript $k$ in $\Pi^{k}(\mathbb{N}_{i})$ is omitted when $k=1$). In this set, we use $v_{j}^{k,i}\in\{1,...,d_{i}\}$ to denote the index of element $j$ in the permutation set $\Pi^{k}(\mathbb{N}_{i})$. To elucidate the notations, we give the following example.
\begin{figure}[bht]
\centering
\includegraphics[width=0.24\textwidth]{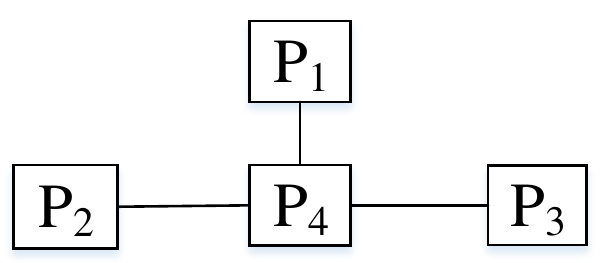}
\caption{Interconnection of a large-scale system with four subsystems. }
\label{figure98}
\end{figure}

\begin{example}
Suppose there is a large-scale system in Fig.~\ref{figure98} with $4$ subsystems, where $\mathbb{N}=\{1,2,3,4\}$ and $\mathbb{N}_{4}=\{1,2,3\}:=\{4_{1},4_{2},4_{3}\}$. The shift permutation operator $\Pi$ defined on $\mathbb{N}_{4}$ is given by:
\begin{align*}
  \Pi\{4_{1},4_{2},4_{3}\}=\{4_{3},4_{1},4_{2}\},\\
  \Pi^{2}\{4_{1},4_{2},4_{3}\}=\{4_{2},4_{3},4_{1}\}.
\end{align*}
In this case, one has
\begin{align*}
v_{1}^{1,4}&=2,~v_{2}^{1,4}=3,~v_{3}^{1,4}=1,\\
v_{1}^{2,4}&=3,~v_{2}^{2,4}=1,~v_{3}^{2,4}=2.
\end{align*}
\end{example}

 Round-Robin communication protocol can be described by first applying the operator $\Pi$ to the neighbor set at each instant
$t_{k}=k\Delta$ ($\Delta$ is a constant sampling period, $k=0,1,...)$, and then selecting the first element from the resulting permutation set $\Pi^{k}(\mathbb{N}_{i})$ for updating feedback. Information from the selected neighbor will be used and updated until this neighbor is polled next time. Information from the unselected neighbors remain constant. Therefore, there is a time interval between polling of the same neighbor, which is denoted by $\tau_{i}$ as
\begin{align} \label{102}
\tau_{i}=t_{k+d_{i}}-t_{k}=d_{i}\Delta
,~ i\in \mathbb{N}.
\end{align}

\begin{remark}
For the $i$th sub-controller, the symbol $\tau_{i}$ is considered as a time delay in communication with the same neighbor. Similar ideas can be seen in \cite{liu2012stability} for the analysis of stability and $\mathcal{L}_{2}$-gain for networked control systems, and in \cite{ugrinovskii2014Round} for distributed estimation problems.
\end{remark}

For $t\in[t_{k},t_{k+1})$, the distributed controller to be designed is of the following form:
\begin{align}\label{2}
u_{i}(t)=K_{ii}x_{i}(t)+\sum_{j\in\Pi^{k}(\mathbb{N}_{i})}K_{ij}x_{j}(t_{k-v_{j}^{k,i}+1}),~i\in \mathbb{N},
\end{align}
where $K_{ii}$ and $K_{ij}$ are controller gains to be designed.

\begin{remark}
Each sub-controller \eqref{2} generates its control input by using local information and information from neighbors, which is similar as those in \cite{chen2016distributed,massioni2009distributed,wu2014distributed}. However, unlike these references, only one neighbor is polled at each instant in \eqref{2} under Round-Robin communication protocol.
\end{remark}

Substituting \eqref{2} into \eqref{1} leads to the closed-loop system $(t\in[t_{k},t_{k+1}))$:
\begin{align}\label{4}
\left\{
  \begin{array}{ll}
 \dot{x}_{i}(t)=(A_{ii}+B_{i}K_{ii})x_{i}(t)+\sum\limits_{j\in \mathbb{N}_{i}}A_{ij}x_{j}(t)\\
\quad \quad \quad+B_{i}\sum\limits_{j\in\Pi^{k}(\mathbb{N}_{i})}K_{ij}x_{j}(t_{k-v_{j}^{k,i}+1})+E_{i}w_{i}(t)\\
 \quad~~~=\bar{A}_{ii}x_{i}(t)+\bar{A}_{ij} x_{i}^{c}(t)+\bar{K}_{ij}x_{i}^{d}(t)+E_{i}w_{i}(t), \\
z_{i}(t)=C_{i}x_{i}(t)+F_{i}w_{i}(t), ~i\in \mathbb{N},
  \end{array}
\right.
\end{align}
where
\begin{align*}
\bar{A}_{ii}&=A_{ii}+B_{i}K_{ii},~ \bar{A}_{ij}=[A_{ii_{1}}\cdots A_{ii_{d_{i}}}],\\
\bar{K}_{ij}&=B_{i}[K_{ii_{1}}\cdots K_{ii_{d_{i}}}],~
x_{i}^{c}(t)=[x_{i_{1}}^{\top}(t) \cdots x_{i_{d_{i}}}^{\top}(t)]^{\top},\\
x_{i}^{d}(t)&=[x_{i_{1}}^{\top}(t_{k-v_{i_{1}}^{k,i}+1}) \cdots x_{i_{d_{i}}}^{\top}(t_{k-v_{i_{d_{i}}}^{k,i}+1})]^{\top}.
\end{align*}

Our objective is to design the distributed controller \eqref{2} such that the following two requirements  are satisfied:
\begin{enumerate}[(i)]
  \item The closed-loop system \eqref{4} with $w_{i}(t)\equiv0$ $(i\in \mathbb{N})$ is exponentially stable;
  \item Under zero initial conditions, the closed-loop system \eqref{4} has a bounded $\mathcal{L}_{2}$-gain, i.e.,
\begin{equation}\label{6}
\sum_{i=1}^{N}\int_{0}^{\infty} z_{i}^{\top}(t)z_{i}(t)dt \leq\gamma^{2}\sum_{i=1}^{N}\int_{0}^{\infty}w_{i}^{\top}(t)w_{i}(t)dt,
\end{equation}
where $\gamma$ is the prescribed disturbance attenuation level.
\end{enumerate}

Throughout this paper, we will make the following assumption without loss of generality.
\begin{assumption}[\cite{lee2006sufficient,wang2007robust}]
The matrix $B_{i}$ $(i\in \mathbb{N})$ is of full column rank.
\end{assumption}

For each $B_{i}$ $(i\in \mathbb{N})$, there always exists an invertible matrix $T_{i}$ such that
\begin{align}\label{20}
T_{i}B_{i}=\begin{bmatrix}
I\\
0\end{bmatrix}.
\end{align}
The corresponding $T_{i}$ generally is not unique. A special $T_{i}$ can be obtained by
\begin{align*}
T_{i}=\begin{bmatrix}
(B_{i}^{\top}B_{i})^{-1}B_{i}^{\top}\\
[(B_{i}^{\top})^{\perp}]^{\top}
\end{bmatrix},
\end{align*}
where $(B_{i}^{\top})^{\perp}$ denotes a basis for the null space of $B_{i}^{\top}$.

Before proceeding further,  it is necessary to present the following three lemmas.
  \begin{lemma}[\cite{liu2012wirtinger}]
  \label{Lemma1}
For a given matrix $R\geq0$, any differentiable function $z(\cdot)$ in $[a,b]\rightarrow \mathbb{R}^{n}$ and $z(a)=0$, the following inequality holds:
\begin{align*}
\frac{\pi^{2}}{4}\int_{a}^{b}z^{\top}(s)Rz(s)ds\leq (b-a)^{2}\int_{a}^{b}\dot{z}^{\top}(s)R\dot{z}(s)ds.
\end{align*}
\end{lemma}

 \begin{lemma}[\cite{seuret2015stability}]
  \label{Lemma2}
 For a given matrix $R\geq0,$ any differentiable function $x(\cdot)$ in $[a,b]\rightarrow \mathbb{R}^{n},$ the following inequality holds:
\begin{align*}
\int_{a}^{b}\dot{x}^{\top}(s)R\dot{x}(s)ds\geq\frac{1}{b-a}\begin{bmatrix}\Upsilon_{0}\\
\Upsilon_{1}\end{bmatrix}^{\top}\begin{bmatrix}R& \mathbf{0}\\
\mathbf{0 }& 3R\end{bmatrix}\begin{bmatrix}\Upsilon_{0}\\
\Upsilon_{1}\end{bmatrix},
\end{align*}
where
\begin{align*}
\Upsilon_{0}&=x(b)-x(a),\\
\Upsilon_{1}&=x(b)+x(a)-\frac{2}{b-a}\int_{a}^{b}x(s)ds.
\end{align*}
\end{lemma}

\begin{lemma}[\cite{ugrinovskii2014Round}]
\label{Lemma3}
Consider a vector $\delta_{i}(t)=[(\delta_{i}^{0})^{\top}(t),...,(\delta_{i}^{d_{i}})^{\top}(t)]^{\top}$ $(i\in \mathbb{N})$ with $d_{i}\in \mathbb{N}^{+}$, if there exist matrices $\hat{R}_{i}$ $(i\in \mathbb{N})$ and $G_{i}$ with compatible dimensions such that
\begin{align}\label{5}
\begin{bmatrix}
\hat{R}_{i} & G_{i}\\
G_{i}^{\top}  & \hat{R}_{i}
\end{bmatrix}\geq0,
\end{align}
then
\begin{enumerate}[(i)]
  \item $\tau_{i}\Big[\frac{1}{t-t_{k}}(\delta_{i}^{0})^{\top}(t)\hat{R}_{i}\delta_{i}^{0}(t)\\
+\sum_{v=1}^{d_{i}-1}\frac{1}{t_{k-v+1}-t_{k-v}}(\delta_{i}^{v})^{\top}(t)\hat{R}_{i}\delta_{i}^{v}(t) \\
+\frac{1}{t_{k-d_{i}+1}-t+\tau_{i}}(\delta_{i}^{d_{i}})^{\top}(t)\hat{R}_{i}\delta_{i}^{d_{i}}(t)\Big]\\
\geq\delta_{i}^{\top}(t)\Psi_{i}\delta_{i}(t), \quad d_{i}\geq2$;\\
  \item $\tau_{i}\Big[\frac{1}{t-t_{k}}(\delta_{i}^{0})^{\top}(t)\hat{R}_{i}\delta_{i}^{0}(t)
+\frac{1}{t_{k}-t+\tau_{i}}(\delta_{i}^{1})^{\top}(t)\hat{R}_{i}\delta_{i}^{1}(t)\Big]\\
\geq\delta_{i}^{\top}(t)\Psi_{i}\delta_{i}(t), \quad d_{i}=1$;
\end{enumerate}
where
\begin{align}\label{eq1}
\Psi_{i}= \left\{
            \begin{array}{ll}
              \begin{bmatrix}
\hat{R}_{i} & \frac{1}{2}(G_{i}+G_{i}^{\top}) & \cdots & \frac{1}{2}(G_{i}+G_{i}^{\top}) \\
* & \hat{R}_{i} & \cdots & \frac{1}{2}(G_{i}+G_{i}^{\top}) \\
\vdots & \vdots & \ddots & \vdots\\
* & * & \cdots & \hat{R}_{i}
\end{bmatrix}, & d_{i}\geq2; \\
          \begin{bmatrix}
\hat{R}_{i} & \frac{1}{2}(G_{i}+G_{i}^{\top}) \\
* & \hat{R}_{i}
\end{bmatrix}, & d_{i}=1.
            \end{array}
          \right.
\end{align}
\end{lemma}

\section{Distributed Networked Controller Design}
In this section, distributed networked controllers will be designed for achieving stability and the
prescribed $\mathcal{L}_{2}$-gain of the closed-loop system \eqref{4}.
\begin{theorem}\label{Theorem1}
Given positive constants $\Delta,\gamma,\alpha_{i}$ $(i\in \mathbb{N})$, $h_{i}$ satisfying  $0<h_{i}<2\alpha_{i}d_{i}^{-1}$. If there exist matrices $W_{i}>0$ $(i\in \mathbb{N})$ and Lyapunov function $V_{s}(t)=\mathop{\sum}\limits_{i=1}^{N}V_{i}(t)$ such that
 \begin{align}\label{8}
&\nonumber\dot{V}_{i}(t)+2\alpha_{i}V_{i}(t)-\sum_{j\in \mathbb{N}_{i}}h_{j}V_{j}(t)+\sum_{j\in \mathbb{N}_{i}}\tau_{j}^{2}\dot{x}_{i}^{\top}(t)W_{i}\dot{x}_{i}(t)\\
\nonumber&-\sum_{j\in \mathbb{N}_{i}}\mu_{j}(x_{j}(t),x_{j}(t_{k-v_{j}^{k,i}+1}))\\
&+z_{i}^{\top}(t)z_{i}(t)-\gamma^{2}w_{i}^{\top}(t)w_{i}(t) \leq0, ~~t\in[t_{k},t_{k+1}),
\end{align}
hold for all $i\in \mathbb{N}$, where $\tau_{j}=d_{j}\Delta$, and
\begin{align}
\nonumber&\mu_{j}(x_{j}(t),x_{j}(t_{k-v_{j}^{k,i}+1}))\\
\label{111}=&\frac{\pi^{2}}{4}(x_{j}(t)-x_{j}(t_{k-v_{j}^{k,i}+1}))^{\top}W_{j}(x_{j}(t)-x_{j}(t_{k-v_{j}^{k,i}+1})),
\end{align}
then the closed-loop system \eqref{4} is exponentially stable and has $\mathcal{L}_{2}$-gain less than $\gamma$.
\end{theorem}

\begin{proof}
By changing the order of summation in the fourth term in \eqref{8}, we further obtain that
\begin{align}\label{9}
\sum_{i=1}^{N}\sum_{j\in \mathbb{N}_{i}}\tau_{j}^{2}\dot{x}_{i}^{\top}(t)W_{i}\dot{x}_{i}(t)=\sum_{i=1}^{N}\sum_{j\in \mathbb{N}_{i}}\tau_{i}^{2}\dot{x}_{j}^{\top}(t)W_{j}\dot{x}_{j}(t).
\end{align}

Summing up both sides of \eqref{8} from $i=1$ to $N$ and noting the equation \eqref{9}, one has
\begin{align}\label{11}
&\nonumber\sum_{i=1}^{N}\Big[\dot{V_{i}}(t)+2\alpha_{i}V_{i}(t)-\sum_{j\in \mathbb{N}_{i}}h_{j}V_{j}(t)\\
&\nonumber \quad~~+z_{i}^{\top}(t)z_{i}(t)-\gamma^{2}w_{i}^{\top}(t)w_{i}(t)\Big]\\
\leq&\sum_{i=1}^{N}\sum_{j\in \mathbb{N}_{i}}\Big[\mu_{j}(x_{j}(t),x_{j}(t_{k-v_{j}^{k,i}+1}))-\tau_{i}^{2}\dot{x}_{j}^{\top}(t)W_{j}\dot{x}_{j}(t)\Big].
\end{align}

Then, integrating both sides of \eqref{11} from $t=0$ to $T$ leads to
\begin{align}\label{12}
&\nonumber\sum_{i=1}^{N}\int_{0}^{T}(\dot{V_{i}}(t)+2\alpha_{i}V_{i}(t)-\sum_{j\in \mathbb{N}_{i}}h_{j}V_{j}(t))dt\\
\nonumber&+\sum_{i=1}^{N}\int_{0}^{T}(z_{i}^{\top}(t)z_{i}(t)-\gamma^{2}w_{i}^{\top}(t)w_{i}(t))dt\\
\leq&\nonumber\sum_{i=1}^{N}\sum_{j\in \mathbb{N}_{i}}\Big[\int_{0}^{T}\mu_{j}(x_{j}(t),x_{j}(t_{k-v_{j}^{k,i}+1}))dt\\
&\quad \quad \quad \quad -\tau_{i}^{2}\int_{0}^{T}\dot{x}_{j}^{\top}(t)W_{j}\dot{x}_{j}(t)dt\Big].
\end{align}

Assume that $T\in[t_{m},t_{m+1})$, and $m$ is sufficient large such that the inequality $m-ld_{i}-v_{j}^{m,i}+1\geq0$ has positive integer solutions. Let $l$ be the largest positive integer among all solutions. Then, the partition of the interval $[0,T]$ is illustrated in Fig.~\ref{figure9}. Here, the symbol $v_{j}^{m,i}$ denotes the index of $j$ in the permutation set $\Pi^{m}(\mathbb{N}_{i})$.
\begin{figure}[bht]
\centering
\includegraphics[width=0.4\textwidth]{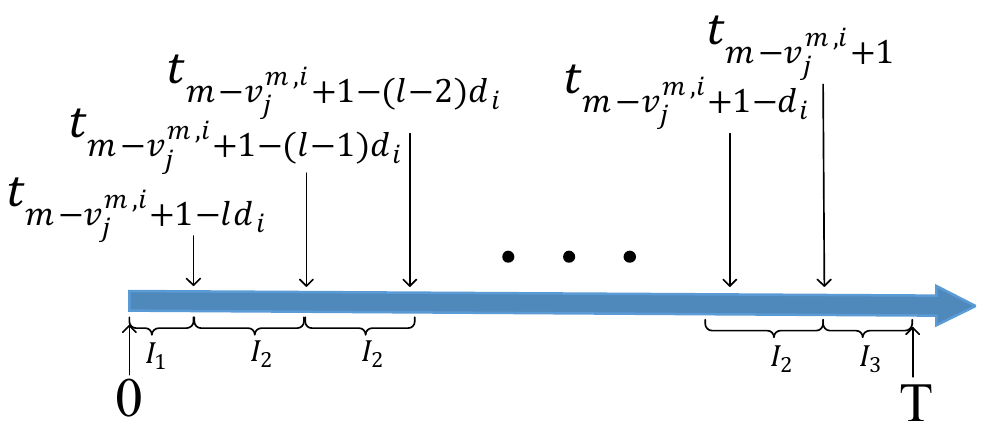}
\caption{The partition of the time interval [0,T].}
\label{figure9}
\end{figure}\\

To be more precisely, the time interval $[0,T]$ is partitioned into the following subintervals:
\begin{align}
\nonumber[0,T]=&[0,t_{m-v_{j}^{m,i}+1-ld_{i}})\\
\nonumber&\cup [t_{m-v_{j}^{m,i}+1-ld_{i}},t_{m-v_{j}^{m,i}+1-(l-1)d_{i}})\\
\nonumber&\cup [t_{m-v_{j}^{m,i}+1-(l-1)d_{i}},t_{m-v_{j}^{m,i}+1-(l-2)d_{i}})\cdots\\
\nonumber &\cup [t_{m-v_{j}^{m,i}+1-d_{i}},t_{m-v_{j}^{m,i}+1})\cup [t_{m-v_{j}^{m,i}+1},T]\\
\nonumber=&\underbrace{[0,t_{m-v_{j}^{m,i}+1-ld_{i}})}_{I_{1}}\\
\nonumber&\cup\large(\cup_{d=1}^{l}\underbrace{[t_{m-v_{j}^{m,i}+1-dd_{i}},t_{m-v_{j}^{m,i}+1-(d-1)d_{i}})}_{I_{2}}\large)\\
\nonumber&\cup\underbrace{[t_{m-v_{j}^{m,i}+1},T]}_{I_{3}}.
\end{align}

It follows from the above partition that the first term in the right hand side of \eqref{12} can be written as:
\begin{align}\label{13}
\nonumber&\int_{0}^{T}\mu_{j}(x_{j}(t),x_{j}(t_{k-v_{j}^{k,i}+1}))dt\\
\nonumber=&\int_{0}^{t_{m-v_{j}^{m,i}+1-ld_{i}}}\mu_{j}(x_{j}(t),x_{j}(t_{k-v_{j}^{k,i}+1}))dt\\
\nonumber &+\sum_{d=1}^{l}\int_{t_{m-v_{j}^{m,i}+1-dd_{i}}}^{t_{m-v_{j}^{m,i}+1-(d-1)d_{i}}}\mu_{j}(x_{j}(t),x_{j}(t_{k-v_{j}^{k,i}+1}))dt\\
&+\int_{t_{m-v_{j}^{m,i}+1}}^{T}\mu_{j}(x_{j}(t),x_{j}(t_{k-v_{j}^{k,i}+1}))dt.
\end{align}

Note that under Round-Robin communication protocol, each sub-controller requires its neighbor's state information periodically. Therefore, if the $i$th sub-controller polls the $j$th $(j\in \mathbb{N}_{i})$ sub-controller at time $t_{m-v_{j}^{m,i}+1-dd_{i}}$ $(d\in \{1,2,...,l\})$, the next time the same neighbor will be polled at $t_{m-v_{j}^{m,i}+1-(d-1)d_{i}}$. During the time interval $(t_{m-v_{j}^{m,i}+1-dd_{i}},t_{m-v_{j}^{m,i}+1-(d-1)d_{i}})$, the $i$th sub-controller requires the other neighbor's information at instants, information from the $j$th sub-controller remain constant. That is, for $t_{k-v_{j}^{k,i}+1}\in (t_{m-v_{j}^{m,i}+1-dd_{i}},t_{m-v_{j}^{m,i}+1-(d-1)d_{i}})$, one has
\begin{align*}
x_{j}(t_{k-v_{j}^{k,i}+1})\equiv x_{j}(t_{m-v_{j}^{m,i}+1-dd_{i}}),
\end{align*}
which means that
\begin{align}\label{914}
 \mu_{j}(x_{j}(t),x_{j}(t_{k-v_{j}^{k,i}+1})) \equiv \mu_{j}(x_{j}(t),x_{j}(t_{m-v_{j}^{m,i}+1-dd_{i}})).
\end{align}

Consider the second integrating term in the right hand side of \eqref{13} $(d\in \{1,2,...,l\})$, one has
\begin{align}\label{14}
\nonumber&\int_{t_{m-v_{j}^{m,i}+1-dd_{i}}}^{t_{m-v_{j}^{m,i}+1-(d-1)d_{i}}}\mu_{j}(x_{j}(t),x_{j}(t_{k-v_{j}^{k,i}+1}))dt\\
\nonumber=&\int_{t_{m-v_{j}^{m,i}+1-dd_{i}}}^{t_{m-v_{j}^{m,i}+1-(d-1)d_{i}}}\mu_{j}(x_{j}(t),x_{j}(t_{m-v_{j}^{m,i}+1-dd_{i}}))dt\\
\nonumber\leq&\int_{t_{m-v_{j}^{m,i}+1-dd_{i}}}^{t_{m-v_{j}^{m,i}+1-(d-1)d_{i}}}
(d_{i}\Delta)^{2}
\dot{x}_{j}^{\top}(t)W_{j}\dot{x}_{j}(t) dt\\
=&\tau_{i}^{2}\int_{t_{m-v_{j}^{m,i}+1-dd_{i}}}^{t_{m-v_{j}^{m,i}+1-(d-1)d_{i}}}
\dot{x}_{j}^{\top}(t)W_{j}\dot{x}_{j}(t) dt.
\end{align}
Here, the first ``$=$'' holds because of \eqref{914}, the second ``='' holds due to \eqref{102}. Note that the value of $x_{j}(t_{m-v_{j}^{m,i}+1-dd_{i}})$ is a constant over the integral interval, which implies that $\dot{x}_{j}(t_{m-v_{j}^{m,i}+1-dd_{i}})=0$, therefore, the first ``$\leq$'' in \eqref{14} holds because of Lemma \ref{Lemma1}.

It follows from the partition of $[0,T]$ that the interval of $I_{1}=[0,t_{m-v_{j}^{m,i}+1-ld_{i}})$ and $I_{3}=[t_{m-v_{j}^{m,i}+1},T]$ are less than a period $d_{i}\Delta$.
 Therefore, the $j$th $(j\in \mathbb{N}_{i})$ sub-controller can not be polled by the $i$th sub-controller more than one time over the interval $I_{1}$ or $I_{3}$.

During the interval $I_{1}$, information from the $j$th sub-controller remain the same as the initial information at $t=0$. Consider the interval $I_{3}$, information from the $j$th sub-controller update at $t_{m-v_{j}^{m,i}+1}$ and remain constant. That is,
\begin{align*}
x_{j}(t_{k-v_{j}^{k,i}+1})&\equiv x_{j}(0),~~t_{k-v_{j}^{k,i}+1}\in I_{1},\\
x_{j}(t_{k-v_{j}^{k,i}+1})&\equiv x_{j}(t_{m-v_{j}^{m,i}+1}),~~t_{k-v_{j}^{k,i}+1}\in I_{3}.
\end{align*}
We obtain from the similar guideline in \eqref{14} that
\begin{align}\label{15}
\nonumber&\int_{0}^{t_{m-v_{j}^{m,i}+1-ld_{i}}}\mu_{j}(x_{j}(t),x_{j}(t_{k-v_{j}^{k,i}+1}))dt\\
\leq & \tau_{i}^{2}\int_{0}^{t_{m-v_{j}^{m,i}+1-ld_{i}}}\dot{x}_{j}^{\top}(t)W_{j}\dot{x}_{j}(t)dt,
\end{align}
\begin{align}\label{16}
\nonumber&\int_{t_{m-v_{j}^{m,i}+1}}^{T}\mu_{j}(x_{j}(t),x_{j}(t_{k-v_{j}^{k,i}+1}))dt\\
\leq &\tau_{i}^{2}\int_{t_{m-v_{j}^{m,i}+1}}^{T}\dot{x}_{j}^{\top}(t)W_{j}\dot{x}_{j}(t)dt.
\end{align}
Substituting \eqref{14}-\eqref{16} into \eqref{13} leads to
\begin{align}\label{17}
\int_{0}^{T}\mu_{j}(x_{j}(t),x_{j}(t_{k-v_{j}^{k,i}+1}))dt\leq \tau_{i}^{2}\int_{0}^{T}\dot{x}_{j}^{\top}(t)W_{j}\dot{x}_{j}(t)dt.
\end{align}
It follows from \eqref{12} and \eqref{17} that
\begin{align}\label{18}
&\nonumber\sum_{i=1}^{N}\int_{0}^{T}(\dot{V_{i}}(t)+2\alpha_{i}V_{i}(t)-\sum_{j\in \mathbb{N}_{i}}h_{j}V_{j}(t))dt\\
&+\sum_{i=1}^{N}\int_{0}^{T}(z_{i}^{\top}(t)z_{i}(t)-\gamma^{2}w_{i}^{\top}(t)w_{i}(t))dt
\leq0.
\end{align}

Consider the first three terms in \eqref{18}, we obtain that
\begin{align}\label{19}
&\nonumber\sum_{i=1}^{N}\int_{0}^{T}(\dot{V_{i}}(t)+2\alpha_{i}V_{i}(t)-\sum_{j\in \mathbb{N}_{i}}h_{j}V_{j}(t))dt\\
=&\nonumber\int_{0}^{T}(\dot{V}_{s}(t)+\sum_{i=1}^{N}2\alpha_{i}V_{i}(t)-\sum_{i=1}^{N}\sum_{j\in \mathbb{N}_{i}}h_{i}V_{i}(t))dt\\
=&\nonumber\int_{0}^{T}(\dot{V}_{s}(t)+\sum_{i=1}^{N}(2\alpha_{i}-d_{i}h_{i})V_{i}(t))dt\\
\geq&\int_{0}^{T} (\dot{V}_{\rm s}(t)+\varepsilon V_{\rm s}(t))dt,
\end{align}
where $\varepsilon :=\mathop{{\rm min}}\limits_{i\in \mathbb{N}}\{2\alpha_{i}-d_{i}h_{i}\}>0$
since $0<h_{i}<2\alpha_{i}d_{i}^{-1}$.
 The first ``$=$'' in \eqref{19} holds due to the fact
that
  $\sum_{i=1}^{N}\sum_{j\in \mathbb{N}_{i}}h_{j}V_{j}(t)=\sum_{i=1}^{N}\sum_{j\in \mathbb{N}_{i}}h_{i}V_{i}(t)$.

From \eqref{18} and \eqref{19}, one has
\begin{align*}
&\int_{0}^{T} (\dot{V}_{\rm s}(t)+\varepsilon V_{\rm s}(t))dt\\
+&\sum_{i=1}^{N}\int_{0}^{T}(z_{i}^{\top}(t)z_{i}(t)-\gamma^{2}w_{i}^{\top}(t)w_{i}(t))dt \leq 0,
\end{align*}
which implies that
\begin{align}\label{920}
\nonumber&V_{\rm s}(T)+\sum_{i=1}^{N}\int_{0}^{T}z_{i}^{\top}(t)z_{i}(t)dt\\
\leq&e^{-\varepsilon T}V_{\rm s}(0)+\gamma^{2}\sum_{i=1}^{N}\int_{0}^{T}w_{i}^{\top}(t)w_{i}(t)dt.
\end{align}
When $w_{i}(t)\equiv0$ $(i\in \mathbb{N})$, we have $V_{\rm s}(T)\rightarrow 0$ as $T\rightarrow \infty$, which implies that the closed-loop system \eqref{4} is exponentially stable. On the other hand, when $w_{i}(t)\neq0$ $(i\in \mathbb{N})$, the inequality \eqref{6} is satisfied from \eqref{920} under zero initial conditions $(V_{s}(0)=0)$. Thus, the proof is completed, the closed-loop system \eqref{4} is exponentially stable and has $\mathcal{L}_{2}$-gain less than $\gamma$.
\end{proof}

In what follows, sufficient conditions in the form of LMIs will be derived such that \eqref{8} is satisfied for all $i\in \mathbb{N}$. We begin with the definitions of $Y_{i}$ $(i\in \mathbb{N})$, $\xi_{i}(t)$ and $\delta_{i}(t)$ which is convenient for the subsequent use in this paper.
\setlength{\arraycolsep}{1.5pt}{\begin{align*}
Y_{i}:=&\left\{
          \begin{array}{ll}
            \left[
	\begin{array}{c|cccc|c|cccc}
I & -I  & \mathbf{0} & \cdots & \mathbf{0}  & \mathbf{0}   & \mathbf{0} & \mathbf{0} & \cdots & \mathbf{0}\\
I & I  & \mathbf{0}  &  \cdots & \mathbf{0} & \mathbf{0}  & -I & \mathbf{0} & \cdots &\mathbf{0}\\ \hline
\mathbf{0} &I  & -I  &  \cdots & \mathbf{0} & \mathbf{0}   & \mathbf{0} & \mathbf{0} & \cdots & \mathbf{0}\\
\mathbf{0} &I  & I   &   \cdots & \mathbf{0} & \mathbf{0}   & \mathbf{0} & -I  & \cdots &\mathbf{0}\\ \hline
\vdots & \vdots &  \vdots & \ddots & \vdots  & \vdots &   \vdots &  \vdots & \ddots & \vdots\\ \hline
\mathbf{0} &\mathbf{0} &\mathbf{0} &  \cdots &I & -I   & \mathbf{0} & \mathbf{0} & \cdots & \mathbf{0}\\
\mathbf{0} & \mathbf{0} & \mathbf{0}  &   \cdots & I & I   & \mathbf{0} &\mathbf{0} & \cdots & -I
\end{array}
	\right], & d_{i}\geq2; \\
\begin{bmatrix}
I & -I  & \mathbf{0} & \mathbf{0}  & \mathbf{0}\\
I &  I  & \mathbf{0} & -I  & \mathbf{0}\\
\mathbf{0}  & I & -I  & \mathbf{0} &  \mathbf{0}\\
 \mathbf{0} & I &  I  & \mathbf{0} & -I  \\
\end{bmatrix}, & d_{i}=1,
          \end{array}
        \right.\\
\xi_{i}(t):=&\begin{bmatrix}x_{i}^{\top}(t),&
(\xi_{i}^{a})^{\top}(t),&
x_{i}^{\top}(t-\tau_{i}),&
(\xi_{i}^{b})^{\top}(t)\end{bmatrix}^{\top}\\
=&\left\{
   \begin{array}{ll}
     \left[
	\begin{array}{c}
x_{i}(t)\\ \hline
x_{i}(t_{k})\\
\vdots\\
x_{i}(t_{k-d_{i}+1})\\ \hline
x_{i}(t-\tau_{i})\\ \hline
\frac{2}{t-t_{k}}\int_{t_{k}}^{t}x_{i}(s)ds\\
\frac{2}{t_{k}-t_{k-1}}\int_{t_{k-1}}^{t_{k}}x_{i}(s)ds\\
\vdots\\
\frac{2}{t_{k-d_{i}+1}-t+\tau_{i}}\int_{t-\tau_{i}}^{t_{k-d_{i}+1}}x_{i}(s)ds
\end{array}
	\right], & d_{i}\geq 2; \\
     \left[
	\begin{array}{c}
x_{i}(t)\\
x_{i}(t_{k})\\
x_{i}(t-\tau_{i})\\
\frac{2}{t-t_{k}}\int_{t_{k}}^{t}x_{i}(s)ds\\
\frac{2}{t_{k}-t+\tau_{i}}\int_{t-\tau_{i}}^{t_{k}}x_{i}(s)ds
\end{array}
	\right], & d_{i}= 1,
   \end{array}
 \right.\\
\delta_{i}(t):=&\begin{bmatrix}(\delta_{i}^{0})^{\top}(t),&
(\delta_{i}^{1})^{\top}(t),&
\cdots,&
(\delta_{i}^{d_{i}})^{\top}(t)\end{bmatrix}^{\top}\\
=&\left\{
    \begin{array}{ll}
      \left[
	\begin{array}{c}
f_{i}(t,t_{k})\\
f_{i}(t_{k},t_{k-1})\\
\vdots \\
f_{i}(t_{k-d_{i}+1},t-\tau_{i})
\end{array}
	\right], & d_{i}\geq2; \\
 \left[
	\begin{array}{c}
f_{i}(t,t_{k})\\
f_{i}(t_{k},t-\tau_{i})
\end{array}
	\right], & d_{i}=1,
    \end{array}
  \right.
\end{align*}}
where $f_{i}(p,q)$ $(i\in \mathbb{N})$ is a vector function defined as
\begin{align*}
f_{i}(p,q)&=\begin{bmatrix}
x_{i}(p)-x_{i}(q)\\
x_{i}(p)+x_{i}(q)-\frac{2}{p-q}\int_{q}^{p}x_{i}(s)ds
\end{bmatrix}.
\end{align*}
It follows from the above definitions that the following relation holds:
\begin{align}\label{27}
\delta_{i}(t)=Y_{i}\xi_{i}(t),~~
d_{i}\in \mathbb{N}^{+}.
\end{align}
Define $\bar{\Psi}_{i}:=e^{-2\alpha_{i}\tau_{i}}Y_{i}^{\top}\Psi_{i}Y_{i}$ $(i\in \mathbb{N})$ with $\Psi_{i}$ in \eqref{eq1}, and partition $\bar{\Psi}_{i}$ in accordance with the partition of $\xi_{i}(t)=[x_{i}^{\top}(t),(\xi_{i}^{a})^{\top}(t),x_{i}^{\top}(t-\tau_{i}),(\xi_{i}^{b})^{\top}(t)]^{\top}$, we get
\begin{align}\label{930}
\bar{\Psi}_{i}=e^{-2\alpha_{i}\tau_{i}}Y_{i}^{\top}\Psi_{i}Y_{i}=\begin{bmatrix}
\bar{\Psi}_{i}^{11} & \bar{\Psi}_{i}^{12} & \bar{\Psi}_{i}^{13} & \bar{\Psi}_{i}^{14}\\
* & \bar{\Psi}_{i}^{22} & \bar{\Psi}_{i}^{23}& \bar{\Psi}_{i}^{24}\\
* & * & \bar{\Psi}_{i}^{33}& \bar{\Psi}_{i}^{34}\\
* & * & * & \bar{\Psi}_{i}^{44}
\end{bmatrix},~i\in \mathbb{N}.
\end{align}

Now, we are in a position to state the following proposition.
\begin{proposition}
If there exists matrices $G_{i}$ $(i\in \mathbb{N})$ and $\hat{R}_{i}=\begin{bmatrix}R_{i}&\mathbf{0}\\ \mathbf{0}&3R_{i}\end{bmatrix}$ satisfying \eqref{5}, then
\begin{align}\label{927}
\nonumber&-\tau_{i}\int_{t-\tau_{i}}^{t}e^{2\alpha_{i}(s-t)}\dot{x}_{i}^{\top}(s)R_{i}\dot{x}_{i}(s)ds\\
\leq &-\xi_{i}^{\top}(t)\bar{\Psi}_{i}\xi_{i}(t),~~t\in[t_{k},t_{k+1}),
\end{align}
hold for any $d_{i}\in \mathbb{N}^{+}$, where $\bar{\Psi}_{i}$ $(i\in \mathbb{N})$ is given by \eqref{930}.
\end{proposition}
\begin{proof}
The proof is divided into two cases. That is,

\textbf{Case 1:} $d_{i}\geq 2$, one has
 \begin{align}\label{eq26}
\nonumber&-\tau_{i}\int_{t-\tau_{i}}^{t}e^{2\alpha_{i}(s-t)}\dot{x}_{i}^{\top}(s)R_{i}\dot{x}_{i}(s)ds\\
\leq \nonumber& -\tau_{i}e^{-2\alpha_{i}\tau_{i}}\int_{t-\tau_{i}}^{t}\dot{x}_{i}^{\top}(s)R_{i}\dot{x}_{i}(s)ds\\
=\nonumber&-\tau_{i}e^{-2\alpha_{i}\tau_{i}}\Big[\int_{t_{k}}^{t}\dot{x}_{i}^{\top}(s)R_{i}\dot{x}_{i}(s)ds\\
&\nonumber+\sum_{v=1}^{d_{i}-1}\int_{t_{k-v}}^{t_{k-v+1}}\dot{x}_{i}^{\top}(s)R_{i}\dot{x}_{i}(s)ds\\
&\nonumber+\int_{t-\tau_{i}}^{t_{k-d_{i}+1}}\dot{x}_{i}^{\top}(s)R_{i}\dot{x}_{i}(s)ds \Big]\\
\nonumber\leq&-\tau_{i}e^{-2\alpha_{i}\tau_{i}}\Big[\frac{1}{t-t_{k}}(\delta_{i}^{0})^{\top}(t)\hat{R}_{i}\delta_{i}^{0}(t)\\ \nonumber&+\sum_{v=1}^{d_{i}-1}\frac{1}{t_{k-v+1}-t_{k-v}}(\delta_{i}^{v})^{\top}(t)\hat{R}_{i}\delta_{i}^{v}(t) \\
\nonumber&+\frac{1}{t_{k-d_{i}+1}-t+\tau_{i}}(\delta_{i}^{d_{i}})^{\top}(t)\hat{R}_{i}\delta_{i}^{d_{i}}(t)\Big]\\
\nonumber\leq&-e^{-2\alpha_{i}\tau_{i}}\delta_{i}^{\top}(t)\Psi_{i}\delta_{i}(t)\\
\nonumber=&-e^{-2\alpha_{i}\tau_{i}} \xi_{i}^{\top}(t)Y_{i}^{\top}\Psi_{i}Y_{i}\xi_{i}(t)\\
=&-\xi_{i}^{\top}(t)\bar{\Psi}_{i}\xi_{i}(t),~~t\in[t_{k},t_{k+1}).
\end{align}
Here, since $e^{2\alpha_{i}(s-t)}\geq e^{-2\alpha_{i}\tau_{i}}$  for $s\in[t-\tau_{i},t]$, then the first ``$\leq$'' in \eqref{eq26} holds. The first ``$=$'' holds because of the partition of the integral interval $[t-\tau_{i},t]=[t-\tau_{i},t_{k-d_{i}+1})\cup(\cup_{v=1}^{d_{i}-1}[t_{k-v},t_{k-v+1}))\cup[t_{k},t]$. The second $``\leq" $ holds
due to Lemma \ref{Lemma2}, the last $``\leq" $ holds because of Lemma \ref{Lemma3}. The second ``$=$'' holds due to  \eqref{27}, and the last ``$=$'' holds because of \eqref{930}.

\textbf{Case 2:} $d_{i}=1$, it follows from the similar guideline in \eqref{eq26} that
 \begin{align}\label{eq27}
\nonumber&-\tau_{i}\int_{t-\tau_{i}}^{t}e^{2\alpha_{i}(s-t)}\dot{x}_{i}^{\top}(s)R_{i}\dot{x}_{i}(s)ds\\
\leq \nonumber& -\tau_{i}e^{-2\alpha_{i}\tau_{i}}\int_{t-\tau_{i}}^{t}\dot{x}_{i}^{\top}(s)R_{i}\dot{x}_{i}(s)ds\\
=\nonumber&-\tau_{i}e^{-2\alpha_{i}\tau_{i}}\Big[\int_{t_{k}}^{t}\dot{x}_{i}^{\top}(s)R_{i}\dot{x}_{i}(s)ds+\int_{t-\tau_{i}}^{t_{k}}\dot{x}_{i}^{\top}(s)R_{i}\dot{x}_{i}(s)ds \Big]\\
\nonumber\leq&-\tau_{i}e^{-2\alpha_{i}\tau_{i}}\Big[\frac{1}{t-t_{k}}(\delta_{i}^{0})^{\top}(t)\hat{R}_{i}\delta_{i}^{0}(t)\\
\nonumber &+\frac{1}{t_{k}-t+\tau_{i}}(\delta_{i}^{1})^{\top}(t)\hat{R}_{i}\delta_{i}^{1}(t)\Big]\\
\nonumber\leq&-e^{-2\alpha_{i}\tau_{i}}\delta_{i}^{\top}(t)\Psi_{i}\delta_{i}(t)\\
=&-\xi_{i}^{\top}(t)\bar{\Psi}_{i}\xi_{i}(t),~~t\in[t_{k},t_{k+1}).
\end{align}
Therefore, the inequality \eqref{927} holds for any $d_{i}\in \mathbb{N}^{+}$, the proof is completed.
\end{proof}

\begin{remark}
  The term $\int_{t-\tau_{i}}^{t}e^{2\alpha_{i}(s-t)}\dot{x}_{i}^{\top}(s)R_{i}\dot{x}_{i}(s)ds$
also appears
   in \cite{liu2012stability,ugrinovskii2014Round}, where Jensen's inequality is used. However, this paper utilize Lemma \ref{Lemma2}, which is an extension of Jensen's inequality \cite{liu2012wirtinger}, to deal with the integral term to reduce the conservatism of the result.
\end{remark}

\begin{theorem}\label{Theorem2}
Given positive constants $\Delta,\gamma,\alpha_{i}$ $(i\in \mathbb{N})$, $h_{i}$ satisfying $0<h_{i}<2\alpha_{i}d_{i}^{-1}$, the matrix $T_{i}$ satisfying \eqref{20}. If there exist matrices $P_{i}>0$ $(i\in \mathbb{N})$, $Q_{i}\geq0,R_{i}\geq0,W_{i}>0$,\\$G_{i},U_{i},L_{i}^{21},L_{i}^{22},H_{i}^{21},H_{i}^{22},M_{i}^{21},M_{i}^{22},N_{i}^{21},N_{i}^{22},X_{i}$ and $Z_{i}$ such that \eqref{5} and
\setlength{\arraycolsep}{2pt}\begin{align}\label{23}
\Theta_{i}=\begin{bmatrix}
  \Theta_{i}^{11} & \Theta_{i}^{12} & \mathbf{0}  & \mathbf{0}  & \mathbf{0} & \Theta_{i}^{16} &  \Theta_{i}^{17}& \Theta_{i}^{18}\\
              *         & \Theta_{i}^{22} & -\bar{\Psi}_{i}^{12} & -\bar{\Psi}_{i}^{13}  & -\bar{\Psi}_{i}^{14} & \Theta_{i}^{26} & \Theta_{i}^{27}&  \Theta_{i}^{28}\\
 *       & *       & -\bar{\Psi}_{i}^{22} & -\bar{\Psi}_{i}^{23}  & -\bar{\Psi}_{i}^{24} & \mathbf{0} & \mathbf{0}& \mathbf{0}\\
 *         & *      & * &   \Theta_{i}^{44}  & -\bar{\Psi}_{i}^{34} & \mathbf{0} & \mathbf{0}& \mathbf{0}\\
  *                      & *                    & * &    *  & -\bar{\Psi}_{i}^{44} & \mathbf{0} & \mathbf{0}& \mathbf{0}\\
  *                     & *                     & * &  * &  * & \Theta_{i}^{66} & \Theta_{i}^{67}& \Theta_{i}^{68}\\
*                      & *                       & * &  * &  * &  * & \Theta_{i}^{77}& \Theta_{i}^{78}\\
*                      & *                       & * &  * &  * &  * &  *& \Theta_{i}^{88}
\end{bmatrix}<0,
\end{align}
hold for all $i\in \mathbb{N}$,
where
\begin{align*}
\Theta_{i}^{11}&=\tau_{i}^{2}R_{i}-{\rm He}\{L_{i}^{\top}T_{i}\}+\sum_{j\in \mathbb{N}_{i}}\tau_{j}^{2}W_{i},\\
\Theta_{i}^{12}&=P_{i} +L_{i}^{\top}T_{i}A_{ii}+\bar{X}_{i}-T_{i}^{\top}H_{i},\\
\Theta_{i}^{16}&=L_{i}^{\top}T_{i}[A_{ii_{1}}\cdots A_{ii_{d_{i}}}]-T_{i}^{\top}M_{i}, \\
\Theta_{i}^{17}&=\bar{Z}_{i} -T_{i}^{\top}N_{i},~~
\Theta_{i}^{18}=L_{i}^{\top}T_{i}E_{i},\\
\Theta_{i}^{22}&={\rm He}\{\bar{X}_{i}+H_{i}^{\top}T_{i}A_{ii}\} -\bar{\Psi}_{i}^{11}+2\alpha_{i}P_{i}+Q_{i}+C_{i}^{\top}C_{i},\\
\Theta_{i}^{26}&= H_{i}^{\top}T_{i}[A_{ii_{1}} \cdots A_{ii_{d_{i}}}]+A_{ii}^{\top}T_{i}^{\top}M_{i}+\bar{X}^{\top}_{i}, \\
\Theta_{i}^{27}&=\bar{Z}_{i} +A_{ii}^{\top}T_{i}^{\top}N_{i}+ \bar{X}^{\top}_{i},\\
\Theta_{i}^{28}&=H_{i}^{\top}T_{i}E_{i}+C_{i}^{\top}F_{i},~~
\Theta_{i}^{44}=-\bar{\Psi}_{i}^{33}-e^{-2\alpha_{i}\tau_{i}}Q_{i},\\
\Theta_{i}^{66}&=-\Omega_{i}^{1}-\Omega_{i}^{2}+{\rm He}\{M_{i}^{\top}T_{i}[A_{ii_{1}} \cdots A_{ii_{d_{i}}}]\},\\
\Theta_{i}^{67}&=\Omega_{i}^{2}+\bar{Z}_{i} +[A_{ii_{1}} \cdots A_{ii_{d_{i}}}]^{\top}T_{i}^{\top}N_{i},\\
\Theta_{i}^{68}&=M_{i}^{\top}T_{i}E_{i},~~
\Theta_{i}^{77}=-\Omega_{i}^{2}+ {\rm He}\{\bar{Z}_{i}\},\\
\Theta_{i}^{78}&=N_{i}^{\top}T_{i}E_{i},~~
\Theta_{i}^{88}=F_{i}^{\top}F_{i}-\gamma^{2}I,\\
L_{i}&=\begin{bmatrix}
U_{i} & \mathbf{0}\\
L_{i}^{21} &  L_{i}^{22}
\end{bmatrix},~~
H_{i}=\begin{bmatrix}
U_{i} & \mathbf{0}\\
H_{i}^{21} &  H_{i}^{22}
\end{bmatrix},~~\bar{X}_{i}=\begin{bmatrix}
X_{i}\\
\mathbf{0}
\end{bmatrix},\\
M_{i}&=\begin{bmatrix}
U_{i} & \mathbf{0}\\
M_{i}^{21} &  M_{i}^{22}
\end{bmatrix},~~
N_{i}=\begin{bmatrix}
U_{i} & \mathbf{0}\\
N_{i}^{21} &  N_{i}^{22}
\end{bmatrix},~~\bar{Z}_{i}=\begin{bmatrix}
Z_{i}\\
\mathbf{0}
\end{bmatrix},\\
\Omega_{i}^{1}&={\rm diag}(h_{i_{1}}P_{i_{1}},...,h_{i_{d_{i}}}P_{i_{d_{i}}}),\\
\Omega_{i}^{2}&={\rm diag}(
 \frac{\pi^{2}}{4}W_{i_{1}},...,\frac{\pi^{2}}{4}W_{i_{d_{i}}}),\\
\hat{R}_{i}&=\begin{bmatrix}R_{i} & \mathbf{0} \\ \mathbf{0} & 3R_{i}\end{bmatrix},
\end{align*}
the symbol $\bar{\Psi}_{i}^{uv}$ is defined in \eqref{930}, in which $u,v\in\{1,2,3,4\}.$
 Then, the state feedback sub-controllers \eqref{2} with gains
\begin{align}\label{9925}
K_{ii}=(U_{i}^{\top})^{-1}X_{i},\quad [K_{ii_{1}}\cdots K_{ii_{d_{i}}}]=(U_{i}^{\top})^{-1}Z_{i}, \quad i\in \mathbb{N},
\end{align}
ensure that the closed-loop system \eqref{4} is exponentially stable and has $\mathcal{L}_{2}$-gain less than $\gamma$.
\end{theorem}

\begin{proof}
Consider the Lyapunov function $V_{s}(t)=\mathop{\sum}\limits_{i=1}^{N}V_{i}(t)$, where
\begin{align}\label{7}
\nonumber V_{i}(t)=&x_{i}^{\top}(t)P_{i} x_{i}(t) +\int_{t-\tau_{i}}^{t}e^{2\alpha_{i}(s-t)}x_{i}^{\top}(s)Q_{i}x_{i}(s)ds\\ &+\tau_{i}\int_{-\tau_{i}}^{0}\int_{t+\theta}^{t}e^{2\alpha_{i}(s-t)}\dot{x}_{i}^{\top}(s)R_{i}\dot{x}_{i}(s)dsd\theta,~ i\in \mathbb{N}.
\end{align}

Differentiating $V_{i}(t)$ $(i\in \mathbb{N})$ with respect to $t$, we get
\begin{align}\label{9930}
\nonumber\dot{V}_{i}(t)=&2x^{\top}_{i}(t)P_{i} \dot{x}_{i}(t)+x_{i}^{\top}(t)Q_{i}x_{i}(t)+\tau_{i}^{2}\dot{x}_{i}^{\top}(t)R_{i}\dot{x}_{i}(t)\\
\nonumber&-e^{-2\alpha_{i}\tau_{i}}x_{i}^{\top}(t-\tau_{i})Q_{i}x_{i}(t-\tau_{i})\\
\nonumber&-\tau_{i}\int_{t-\tau_{i}}^{t}e^{2\alpha_{i}(s-t)}\dot{x}_{i}^{\top}(s)R_{i}\dot{x}_{i}(s)ds\\
\nonumber&-2\alpha_{i}\int_{t-\tau_{i}}^{t}e^{2\alpha_{i}(s-t)}x_{i}^{\top}(s)Q_{i}x_{i}(s)ds\\
&-2\alpha_{i}\tau_{i}\int_{-\tau_{i}}^{0}\int_{t+\theta}^{t}e^{2\alpha_{i}(s-t)}\dot{x}_{i}^{\top}(s)R_{i}\dot{x}_{i}(s)dsd\theta,~i\in \mathbb{N}.
\end{align}
It follows from \eqref{7} and \eqref{9930} that
\begin{align}\label{28}
\nonumber&\dot{V}_{i}(t)+2\alpha_{i}V_{i}(t)\\
 \nonumber=&2x^{\top}_{i}(t)P_{i}\dot{x}_{i}(t)+\tau_{i}^{2}\dot{x}_{i}^{\top}(t)R_{i}\dot{x}_{i}(t)+x_{i}^{\top}(t)(2\alpha_{i}P_{i} +Q_{i})x_{i}(t) \\
 \nonumber &-e^{-2\alpha_{i}\tau_{i}}x_{i}^{\top}(t-\tau_{i})Q_{i}x_{i}(t-\tau_{i})\\
\nonumber&-\tau_{i}\int_{t-\tau_{i}}^{t}e^{2\alpha_{i}(s-t)}\dot{x}_{i}^{\top}(s)R_{i}\dot{x}_{i}(s)ds\\
\nonumber\leq&2x^{\top}_{i}(t)P_{i}\dot{x}_{i}(t)+\tau_{i}^{2}\dot{x}_{i}^{\top}(t)R_{i}\dot{x}_{i}(t)+x_{i}^{\top}(t)(2\alpha_{i}P_{i} +Q_{i})x_{i}(t) \\
 &-e^{-2\alpha_{i}\tau_{i}}x_{i}^{\top}(t-\tau_{i})Q_{i}x_{i}(t-\tau_{i})- \xi_{i}^{\top}(t)\bar{\Psi}_{i}\xi_{i}(t).
\end{align}
Here, the last ``$\leq$'' holds because of Proposition 1.

On the other hand, from the definition of $V_{i}(t)$ in \eqref{7}, one has
\begin{align} \label{937}
-\sum_{j\in \mathbb{N}_{i}}h_{j}V_{j}(t)\leq -(x_{i}^{c})^{\top}(t)\Omega_{i}^{1}x_{i}^{c}(t),
\end{align}
where $x_{i}^{c}(t)$ is defined after \eqref{4}.

From the definition of  $\mu_{j}(x_{j}(t),x_{j}(t_{k-v_{j}^{k,i}+1}))$ in \eqref{111}, one has
\begin{align}\label{30}
\sum_{j\in \mathbb{N}_{i}}\mu_{j}(x_{j}(t),x_{j}(t_{k-v_{j}^{k,i}+1}))
=\begin{bmatrix}
 x_{i}^{c}(t)\\
 x_{i}^{d}(t)
 \end{bmatrix}^{\top}
 \begin{bmatrix}
\Omega_{i}^{2}  & -\Omega_{i}^{2}\\
 * & \Omega_{i}^{2}
 \end{bmatrix}
 \begin{bmatrix}
 x_{i}^{c}(t)\\
 x_{i}^{d}(t)
 \end{bmatrix},
\end{align}
where $x_{i}^{d}(t)$ is defined after \eqref{4}.

We obtain from \eqref{4} that
\begin{align}\label{22}
\nonumber&(T_{i}^{\top}L_{i}\dot{x}_{i}(t)+T_{i}^{\top}H_{i}x_{i}(t)+T_{i}^{\top}M_{i}x_{i}^{c}(t)+T_{i}^{\top}N_{i}x_{i}^{d}(t))^{\top}\\
\nonumber\times&(-\dot{x}_{i}(t)+\bar{A}_{ii}x_{i}(t)+\bar{A}_{ij} x_{i}^{c}(t)+\bar{K}_{ij}x_{i}^{d}(t)+E_{i}w_{i}(t))\\
\nonumber+&(-\dot{x}_{i}(t)+\bar{A}_{ii}x_{i}(t)+\bar{A}_{ij} x_{i}^{c}(t)+\bar{K}_{ij}x_{i}^{d}(t)+E_{i}w_{i}(t))^{\top}\\
\times& (T_{i}^{\top}L_{i}\dot{x}_{i}(t)+T_{i}^{\top}H_{i}x_{i}(t)+T_{i}^{\top}M_{i}x_{i}^{c}(t)+T_{i}^{\top}N_{i}x_{i}^{d}(t))\equiv0,
\end{align}
where $T_{i}$ $(i\in \mathbb{N})$ is the matrix satisfying \eqref{20}.

Combining \eqref{28}-\eqref{22} with \eqref{8} leads to
\begin{align}\label{936}
\nonumber&\dot{V}_{i}(t)+2\alpha_{i}V_{i}(t)
+\sum_{j\in \mathbb{N}_{i}}\tau_{j}^{2}\dot{x}_{i}^{\top}(t)W_{i}\dot{x}_{i}(t)\\
\nonumber&-\sum_{j\in \mathbb{N}_{i}}h_{j}V_{j}(t)-\sum_{j\in \mathbb{N}_{i}}\mu_{j}(x_{j}(t),x_{j}(t_{k-v_{j}^{k,i}+1}))\\
\nonumber&+z_{i}^{\top}(t)z_{i}(t)-\gamma^{2}w_{i}^{\top}(t)w_{i}(t)\\
\nonumber\leq&2x^{\top}_{i}(t)P_{i}\dot{x}_{i}(t)+\tau_{i}^{2}\dot{x}_{i}^{\top}(t)R_{i}\dot{x}_{i}(t)+x_{i}^{\top}(t)(2\alpha_{i}P_{i} +Q_{i})x_{i}(t) \\
\nonumber&-e^{-2\alpha_{i}\tau_{i}}x_{i}^{\top}(t-\tau_{i})Q_{i}x_{i}(t-\tau_{i})- \xi_{i}^{\top}(t)\bar{\Psi}_{i}\xi_{i}(t)\\
\nonumber&+\sum_{j\in \mathbb{N}_{i}}\tau_{j}^{2}\dot{x}_{i}^{\top}(t)W_{i}\dot{x}_{i}(t)+x_{i}^{\top}(t)C_{i}^{\top}C_{i}x_{i}(t)\\
\nonumber&+2w_{i}^{\top}(t)F_{i}^{\top}C_{i}x_{i}(t)+w_{i}^{\top}(t)(F_{i}^{\top}F_{i}-\gamma^{2}I)w_{i}(t)\\
\nonumber&-\begin{bmatrix}
 x_{i}^{c}(t)\\
 x_{i}^{d}(t)
 \end{bmatrix}^{\top}
 \begin{bmatrix}
\Omega_{i}^{1}+\Omega_{i}^{2}  & -\Omega_{i}^{2}\\
 * & \Omega_{i}^{2}
 \end{bmatrix}
 \begin{bmatrix}
 x_{i}^{c}(t)\\
 x_{i}^{d}(t)
 \end{bmatrix}\\
\nonumber&+(T_{i}^{\top}L_{i}\dot{x}_{i}(t)+T_{i}^{\top}H_{i}x_{i}(t)+T_{i}^{\top}M_{i}x_{i}^{c}(t)+T_{i}^{\top}N_{i}x_{i}^{d}(t))^{\top}\\
\nonumber&\times(-\dot{x}_{i}(t)+\bar{A}_{ii}x_{i}(t)+\bar{A}_{ij} x_{i}^{c}(t)+\bar{K}_{ij}x_{i}^{d}(t)+E_{i}w_{i}(t))\\
\nonumber&+(-\dot{x}_{i}(t)+\bar{A}_{ii}x_{i}(t)+\bar{A}_{ij} x_{i}^{c}(t)+\bar{K}_{ij}x_{i}^{d}(t)+E_{i}w_{i}(t))^{\top}\\
\nonumber&\times (T_{i}^{\top}L_{i}\dot{x}_{i}(t)+T_{i}^{\top}H_{i}x_{i}(t)+T_{i}^{\top}M_{i}x_{i}^{c}(t)+T_{i}^{\top}N_{i}x_{i}^{d}(t))\\
=&\eta_{i}^{\top}(t)\bar{\Theta}_{i}\eta_{i}(t),
\end{align}
where $\eta_{i}(t)=[\dot{x}_{i}^{\top}(t),x_{i}^{\top}(t), (\xi_{i}^{a})^{\top}(t),x_{i}^{\top}(t-\tau_{i}),(\xi_{i}^{b})^{\top}(t),\\(x_{i}^{c})^{\top}(t),
(x_{i}^{d})^{\top}(t),w_{i}^{\top}(t)]^{\top}$,
\setlength{\arraycolsep}{1.5pt}\begin{align*}
\bar{\Theta}_{i}&= \begin{bmatrix}
\Theta_{i}^{11} & \bar{\Theta}_{i}^{12} & \mathbf{0}                  & \mathbf{0}  & \mathbf{0} & \Theta_{i}^{16} &  \bar{\Theta}_{i}^{17} & \Theta_{i}^{18} \\
              *    & \bar{\Theta}_{i}^{22} & -\bar{\Psi}_{i}^{12} & -\bar{\Psi}_{i}^{13}  & -\bar{\Psi}_{i}^{14} & \bar{\Theta}_{i}^{26} & \bar{\Theta}_{i}^{27}& \Theta_{i}^{28} \\
 *     & *        & -\bar{\Psi}_{i}^{22} & -\bar{\Psi}_{i}^{23}  & -\bar{\Psi}_{i}^{24} & \mathbf{0} & \mathbf{0}& \mathbf{0} \\
 *                      & *                    & * &   \Theta_{i}^{44}  & -\bar{\Psi}_{i}^{34} & \mathbf{0} & \mathbf{0}& \mathbf{0} \\
  *        & *                    & * &    *  & -\bar{\Psi}_{i}^{44} & \mathbf{0} & \mathbf{0}& \mathbf{0}\\
  *                     & *                     & * &  * &  * & \Theta_{i}^{66} & \bar{\Theta}_{i}^{67}& \Theta_{i}^{68}\\
*                      & *                       & * &  * &  * &  * & \bar{\Theta}_{i}^{77}& \Theta_{i}^{78}\\
*                      & *                       & * &  * &  * &  * & *& \Theta_{i}^{88}\\
\end{bmatrix},\\
\bar{\Theta}_{i}^{12}&=P_{i} +L_{i}^{\top}T_{i}(A_{ii}+B_{i}K_{ii})-T_{i}^{\top}H_{i},\\
\bar{\Theta}_{i}^{17}&=L_{i}^{\top}T_{i}B_{i}[K_{ii_{1}}\cdots K_{ii_{d_{i}}}]-T_{i}^{\top}N_{i},\\
\bar{\Theta}_{i}^{22}&={\rm He}\{H_{i}^{\top}T_{i}(A_{ii}+B_{i}K_{ii})\}-\bar{\Psi}_{i}^{11}\\
&~~~+2\alpha_{i}P_{i}+Q_{i}+C_{i}^{\top}C_{i},\\
\bar{\Theta}_{i}^{26}&= H_{i}^{\top}T_{i}[A_{ii_{1}} \cdots A_{ii_{d_{i}}}]+(A_{ii}+B_{i}K_{ii})^{\top}T_{i}^{\top}M_{i},\\
\bar{\Theta}_{i}^{27}&= H_{i}^{\top}T_{i}B_{i}[K_{ii_{1}} \cdots K_{ii_{d_{i}}}]+(A_{ii}+B_{i}K_{ii})^{\top}T_{i}^{\top}N_{i},\\
\bar{\Theta}_{i}^{67}&=\Omega_{i}^{2}+M_{i}^{\top}T_{i}B_{i}[K_{ii_{1}} \cdots K_{ii_{d_{i}}}]\\
&~~~+[A_{ii_{1}} \cdots A_{ii_{d_{i}}}]^{\top}T_{i}^{\top}N_{i},\\
\bar{\Theta}_{i}^{77}&=-\Omega_{i}^{2}+ {\rm He}\{N_{i}^{\top}T_{i}B_{i}[K_{ii_{1}} \cdots K_{ii_{d_{i}}}]\}.
\end{align*}

From $\Theta_{i}<0$ in \eqref{23}, we can deduce that $\Theta_{i}^{11}<0$, which implies that ${\rm He}\{L_{i}^{\top}T_{i}\}>0$. Note the matrix $T_{i}$ is invertible, we obtain that $L_{i}$ is invertible. Then, the invertibility of $L_{i}$ implies that $U_{i}$ is invertible due to the structure of $L_{i}$. It follows from \eqref{20}, \eqref{9925} and the structure of $L_{i},\bar{X}_{i}$ that

\begin{align}
\nonumber L_{i}^{\top}T_{i}B_{i}K_{ii}=&\begin{bmatrix}
U_{i}^{\top} & (L_{i}^{21})^{\top}\\
\mathbf{0} & (L_{i}^{22})^{\top}
\end{bmatrix}\begin{bmatrix}
I\\
\mathbf{0}
\end{bmatrix}K_{ii}=\begin{bmatrix}
U_{i}^{\top}K_{ii}\\
\mathbf{0}
\end{bmatrix}\\
\label{31}=&\begin{bmatrix}
X_{i}\\
\mathbf{0}
\end{bmatrix}=\bar{X}_{i},~~i\in \mathbb{N}.
\end{align}
\begin{align}
\nonumber L_{i}^{\top}T_{i}B_{i}[K_{ii_{1}}\cdots K_{ii_{d_{i}}}]=&\begin{bmatrix}
U_{i}^{\top} & (L_{i}^{21})^{\top}\\
\mathbf{0} & (L_{i}^{22})^{\top}
\end{bmatrix}\begin{bmatrix}
I\\
\mathbf{0}
\end{bmatrix}[K_{ii_{1}}\cdots K_{ii_{d_{i}}}]\\
\nonumber=&\begin{bmatrix}
U_{i}^{\top}[K_{ii_{1}}\cdots K_{ii_{d_{i}}}]\\
\mathbf{0}
\end{bmatrix}=\begin{bmatrix}
Z_{i}\\
\mathbf{0}
\end{bmatrix}\\
\label{32}=&\bar{Z}_{i},~~i\in \mathbb{N}.
\end{align}
Similarly, the following relations also hold for all $i\in \mathbb{N}$,
\begin{align}
H_{i}^{\top}T_{i}B_{i}K_{ii}&=\bar{X}_{i}, ~~H_{i}^{\top}T_{i}B_{i}[K_{ii_{1}}\cdots K_{ii_{d_{i}}}]=\bar{Z}_{i}\\
M_{i}^{\top}T_{i}B_{i}K_{ii}&=\bar{X}_{i}, ~~M_{i}^{\top}T_{i}B_{i}[K_{ii_{1}}\cdots K_{ii_{d_{i}}}]=\bar{Z}_{i},\\
\label{eq40} N_{i}^{\top}T_{i}B_{i}K_{ii}&=\bar{X}_{i}, ~~N_{i}^{\top}T_{i}B_{i}[K_{ii_{1}}\cdots K_{ii_{d_{i}}}]=\bar{Z}_{i}.
\end{align}

Substituting \eqref{31}-\eqref{eq40} into $\bar{\Theta}_{i}$ leads to
\begin{align}\label{9300}
\bar{\Theta}_{i}=\Theta_{i}<0.
\end{align}
Here, the last ``$<$'' holds because of \eqref{23}.

It follows from \eqref{936} and \eqref{9300} that
\begin{align}
\nonumber&\dot{V}_{i}(t)+2\alpha_{i}V_{i}(t)
+\sum_{j\in \mathbb{N}_{i}}\tau_{j}^{2}\dot{x}_{i}^{\top}(t)W_{i}\dot{x}_{i}(t)\\
\nonumber&-\sum_{j\in \mathbb{N}_{i}}h_{j}V_{j}(t)-\sum_{j\in \mathbb{N}_{i}}\mu_{j}(x_{j}(t),x_{j}(t_{k-v_{j}^{k,i}+1}))\\
\nonumber&+z_{i}^{\top}(t)z_{i}(t)-\gamma^{2}w_{i}^{\top}(t)w_{i}(t)<0.
\end{align}
 Therefore, the closed-loop system \eqref{4} is exponentially stable and has $\mathcal{L}_{2}$-gain less than $\gamma$ according to Theorem \ref{Theorem1}.
\end{proof}

\begin{remark}
Note the equation \eqref{22}, the addition of the left hand side of \eqref{22} into the right hand side of ``$\leq$" in \eqref{936} does not change the ``$\leq$" in \eqref{936}. However, by applying this technique, we have introduced four auxiliary matrices $L_{i},H_{i},M_{i},N_{i}$ $(i\in \mathbb{N})$ to reduce the conservatism of the results in Theorem \ref{Theorem2}. Similar techniques also can be seen in \cite{ugrinovskii2014Round,liu2012wirtinger}.
\end{remark}

\begin{remark}\label{remark4}
Comparing with the literatures \cite{massioni2009distributed},\cite{massioni2014distributed,hoffmann2013distributed,wu2017cooperative}, Theorem 1 and Theorem 2 can be applied to deal with both homogeneous and heterogeneous systems. For instance, only the homogeneous cases were studied in \cite{massioni2009distributed}, and some special heterogeneous cases, such as, ``$\alpha$-heterogeneous systems'', ``decomposable systems'' and multi-agent systems were concerned in \cite{massioni2014distributed},\cite{hoffmann2013distributed} and \cite{wu2017cooperative}, respectively.
\end{remark}

\begin{remark}
The minimum value of the $\mathcal{L}_{2}$-gain $\gamma$ is of interest in many applications, it can be obtained by solving the following optimization problem:
\begin{align}\label{1000}
\begin{split}
\min & \quad \gamma^{2}\\
{\rm subject ~ to} & \quad  \eqref{5},~\eqref{23}.
\end{split}
\end{align}
\end{remark}

\begin{remark}
The distributed controller gains can be obtained by solving a set of LMIs of Theorem \ref{Theorem2} off-line. This requires some centralized information such as ``$A_{ij},P_{i},Q_{i},R_{i}$'' $(i\in \mathbb{N},j\in \mathbb{N}_{i})$. Nonetheless, once the controller gains are designed, the implementation is fully distributed. Each sub-controller only requires local information and information from neighbors to form its control input.
\end{remark}

It should be noted that for each $B_{i}$ $(i\in \mathbb{N})$, there may exist different choices of $T_{i}$ satisfying \eqref{20}. The following theorem shows that the feasibility of the conditions of Theorem \ref{Theorem2} is independent of the choices of $T_{i}$.

\begin{theorem}\label{Theorem3}
If the LMIs conditions \eqref{23} of Theorem \ref{Theorem2} are feasible for some $T_{i}$ satisfying \eqref{20}, then they are feasible for any $\hat{T}_{i}$ satisfying \eqref{20}.
\end{theorem}
\begin{proof}
Since $T_{i}$ and $\hat{T}_{i}$ satisfy \eqref{20}, we obtain that
\begin{align}\label{37}
 \begin{bmatrix}
I\\ \mathbf{0}\end{bmatrix}=T_{i}\hat{T}_{i}^{-1}\begin{bmatrix}
I\\ \mathbf{0}\end{bmatrix}.
\end{align}
 Denote
\begin{align*}
J_{i}=T_{i}\hat{T}_{i}^{-1}=\begin{bmatrix}J_{i}^{11} & J_{i}^{12}\\
J_{i}^{21} & J_{i}^{22}
\end{bmatrix},
\end{align*}
it follows from \eqref{37} that $J_{i}^{11}=I$ and $J_{i}^{21}=\mathbf{0}.$ Consider
\begin{align}
\nonumber L_{i}^{\top}T_{i}&=\begin{bmatrix}
U_{i}^{\top} & (L_{i}^{21})^{\top}\\
\mathbf{0} & (L_{i}^{22})^{\top}
\end{bmatrix}J_{i}\hat{T}_{i}=\begin{bmatrix}
U_{i}^{\top} & (L_{i}^{21})^{\top}\\
\mathbf{0} & (L_{i}^{22})^{\top}
\end{bmatrix}\begin{bmatrix}
I & J_{i}^{12}\\
\mathbf{0} & J_{i}^{22}
\end{bmatrix}\hat{T}_{i}\\
\nonumber&=\begin{bmatrix}
U_{i}^{\top} & (\hat{L}_{i}^{21})^{\top}\\
\mathbf{0} & (\hat{L}_{i}^{22})^{\top}
\end{bmatrix}\hat{T}_{i},
\end{align}
where $(\hat{L}_{i}^{21})^{\top}=U_{i}^{\top}J_{i}^{12}+(L_{i}^{21})^{\top}J_{i}^{22}$, $(\hat{L}_{i}^{22})^{\top}=(L_{i}^{22})^{\top}J_{i}^{22}$.
Then we get $L_{i}^{\top}T_{i}=\hat{L}_{i}^{\top}\hat{T}_{i}$ with
\begin{align*}
\hat{L}_{i}=\begin{bmatrix}
U_{i} & \mathbf{0} \\
 \hat{L}_{i}^{21}&  \hat{L}_{i}^{22}
\end{bmatrix}.
\end{align*}

Similarly, the matrices $\hat{H}_{i}$ $(i\in \mathbb{N})$, $\hat{M}_{i}$ and $\hat{N}_{i}$ also can be found such that \begin{align*}
H_{i}^{\top}T_{i}=\hat{H}_{i}^{\top}\hat{T}_{i},~~M_{i}^{\top}T_{i}=\hat{M}_{i}^{\top}\hat{T}_{i},~~ N_{i}^{\top}T_{i}=\hat{N}_{i}^{\top}\hat{T}_{i}.
\end{align*}
Note that the matrix $T_{i}$ $(i\in \mathbb{N})$ only appears in the
cross terms ``$L_{i}^{\top}T_{i},H_{i}^{\top}T_{i},M_{i}^{\top}T_{i},N_{i}^{\top}T_{i}$'' or their transpositions in \eqref{23}. Therefore, if \eqref{23} is feasible for $T_{i}$ satisfying \eqref{20}, then \eqref{23} is feasible for any $\hat{T}_{i}$ satisfying \eqref{20}.
\end{proof}

\section{Numerical Examples}

In this section, the effectiveness of the proposed method are demonstrated by three numerical examples.

 \begin{example}\label{example2}
This example is used to compare our distributed controller \eqref{2} with the controller in \cite{ghadami2013decomposition} under different system parameters.
\begin{figure}[bht]
\centering
\includegraphics[width=0.35\textwidth]{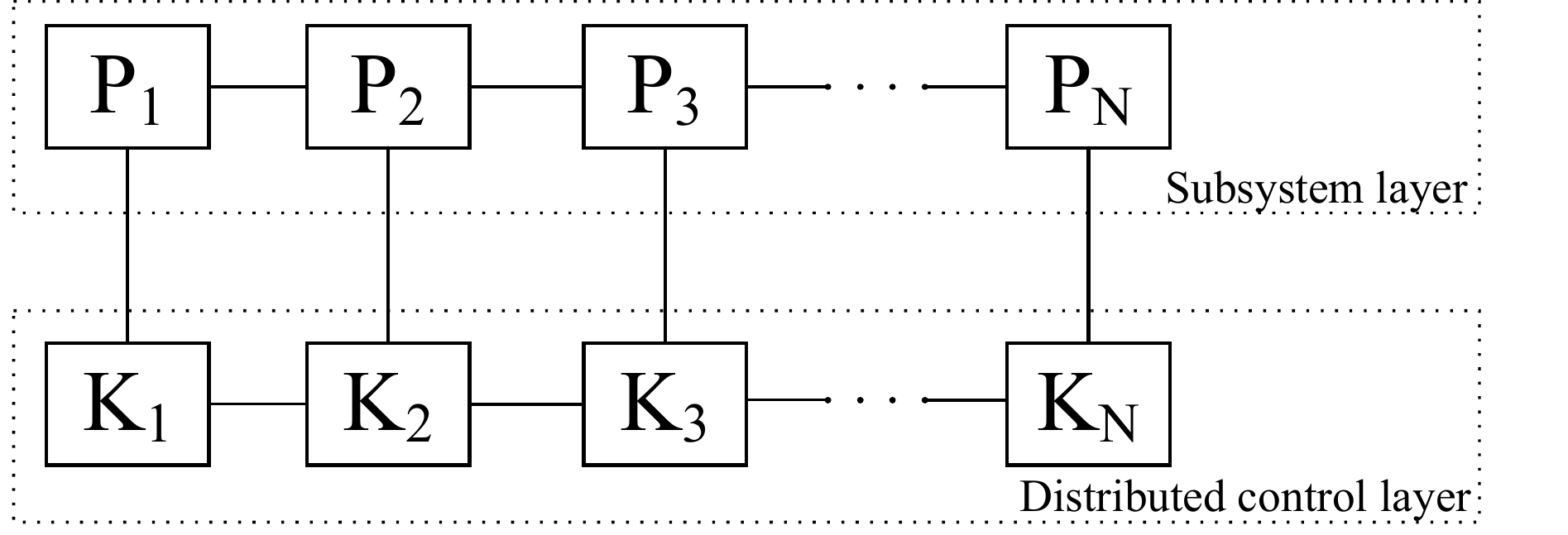}
\caption{Interconnection of the large-scale system.}
\label{figure1}
  \end{figure}

Consider a large-scale system in Fig.~\ref{figure1} with $N=10$ subsystems. The dynamics of the $i$th $(i\in \mathbb{N})$ subsystem is described by equation \eqref{1}, where
\begin{align}
\label{eq43} A_{a}&=\begin{bmatrix}   
-0.7+a  & -0.1\\
0  & -0.8+0.1a
  \end{bmatrix},\quad
 A_{ii}=A_{a}+d_{i} A_{b} , \\
\label{eq44}  A_{b}&=\begin{bmatrix}   
 -0.2 & -0.1+0.2a\\
  0 & -0.1
  \end{bmatrix},\quad A_{ij}=-A_{b},\quad F_{i}=2,\\
\label{eq45}B_{i}&=\begin{bmatrix}
-0.4\\
0.1
  \end{bmatrix},
  \quad
 E_{i}=\begin{bmatrix}   
 0.1\\
  -0.2
 \end{bmatrix}, \quad
C_{i} =\begin{bmatrix}
0.1 &  0
  \end{bmatrix}.
\end{align}
Here, the symbol $a$ is a parameterised scalar, which will
take different values in $\{-0.4,-0.35,...,0.4\}$ for comparison purpose.

This large-scale system is a ``decomposable system'' \cite{ghadami2013decomposition} when the graph Laplacian matrix $L$ is chosen as the ``pattern matrix'' \cite{ghadami2013decomposition}, where
\begin{equation*}
L=\begin{bmatrix}   
    1  &   -1   &  0 &    \cdots &    0 & 0\\
    -1 &    2   & -1 &   \cdots &   0  & 0\\
    \vdots &    \vdots   & \vdots &  \ddots & \vdots & \vdots\\
    0 &     0   & 0&      \cdots  & -1 &  1
  \end{bmatrix}.
\end{equation*}

Now we compare the value $\gamma_{\rm min}$ achieved by the distributed controller \eqref{2} and the controller in \cite{ghadami2013decomposition}. In the optimization problem \eqref{1000}, we choose $\Delta=0.0005$, $h_{i}=0.1$ and $\alpha_{i}=0.4$ $(i\in \mathbb{N})$ in the constraint \eqref{23}. By solving the problem \eqref{1000} with 20 LMIs restrictions
, the comparison result is shown in Fig.~\ref{figure2}.

As Fig.~\ref{figure2} illustrates, the $\mathcal{L}_{2}$-gain $\gamma_{\rm min}$ achieved by the distributed controller \eqref{2} increased slightly compared with that by the controller in \cite{ghadami2013decomposition} when
$a\in \{-0.4,-0.35,...,0.4\}$. However, the distributed controller \eqref{2} leads to about 50$\%$ bandwidth savings, because each sub-controller interacts with only one neighbor at each instant under Round-Robin communication protocol.
\begin{figure}[bht]
\centering
\includegraphics[width=0.46\textwidth]{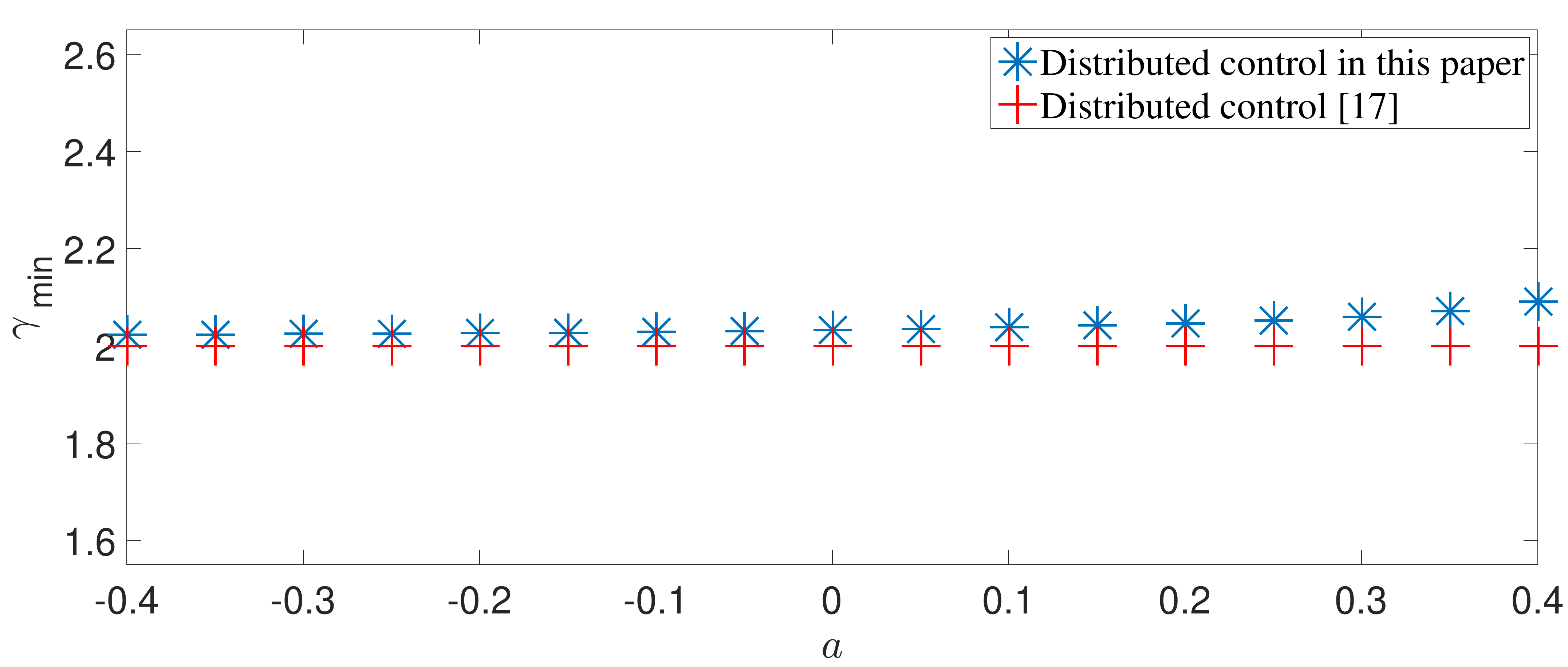}
\caption{$\mathcal{L}_{2}$-gain with respect to the
parameter $a$.}
\label{figure2}
\end{figure}
 \end{example}

\begin{example}
This example is used to compare our distributed controller \eqref{2} with the controller in \cite{ghadami2013decomposition} under different number of subsystems.

Consider a large-scale system in Fig.~\ref{figure1} with $N$ subsystems. The system parameters are given in \eqref{eq43}-\eqref{eq45} with $a=0$, the values of $\alpha_{i}$ $(i\in \mathbb{N})$, $h_{i}$ and
$\Delta$ are chosen the same as those in Example \ref{example2}. In order to obtain the minimal $\mathcal{L}_{2}$-gain $\gamma_{\rm min}$ achieved by our controllers, we need to solve the optimization problem \eqref{1000} with $2N$ LMIs restrictions. The comparison result is shown in Fig.~\ref{figure3}.
\begin{figure}[bht]
\centering
\includegraphics[width=0.46\textwidth]{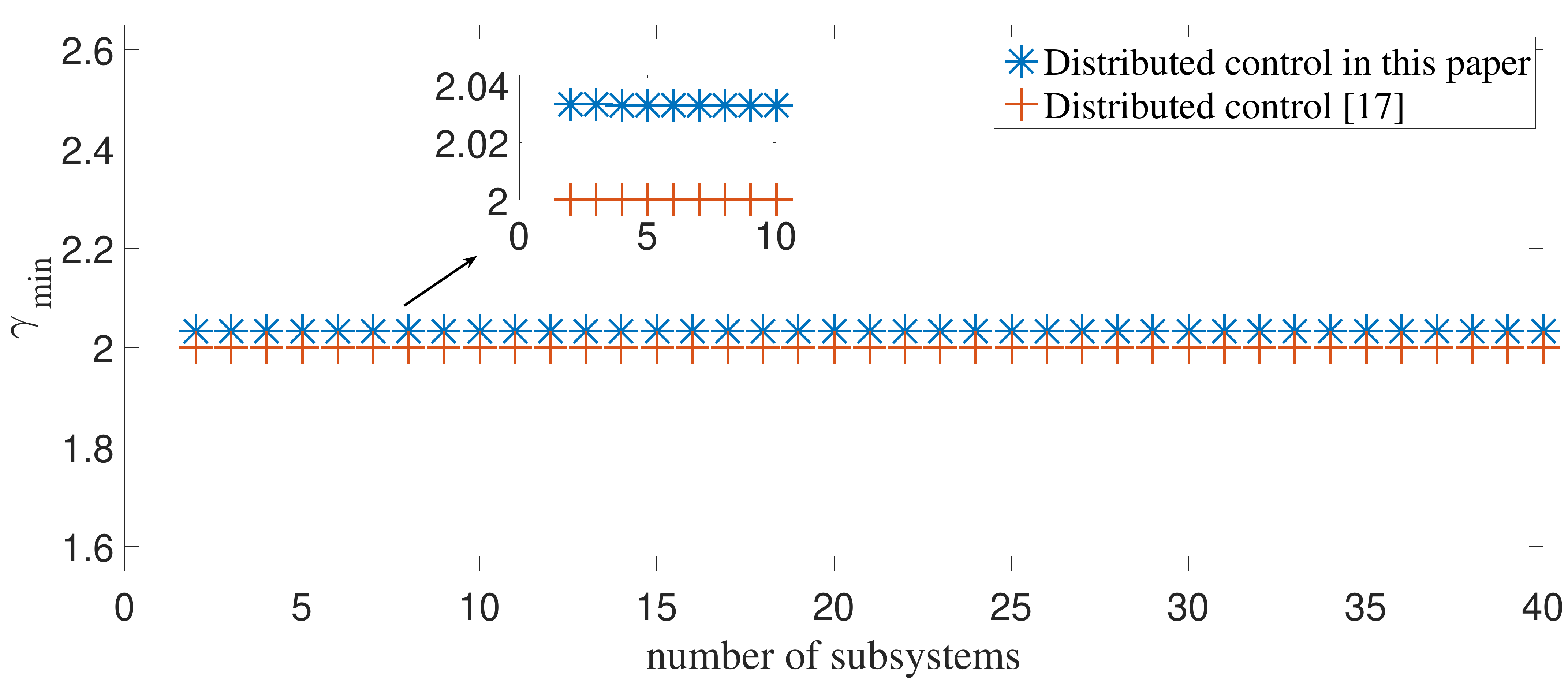}
\caption{$\mathcal{L}_{2}$-gain for different number of subsystems.}
\label{figure3}
\end{figure}

From Fig.~\ref{figure3}, it is interesting to observe that the minimal $\mathcal{L}_{2}$-gain $\gamma_{\rm min}$ achieved by the controller \eqref{2} changes slightly under different number of subsystems. The reason is that this example considers identical subsystems and interconnections.

The $\mathcal{L}_{2}$-gain $\gamma_{\rm min}$ achieved by the controller \eqref{2} is about $\gamma_{\rm min}$=2.0311 under different number of subsystems, while the value achieved by the controller in \cite{ghadami2013decomposition} is about $\gamma_{\rm min}$=2.0001. The value of $\gamma_{\rm min}$
achieved by our controllers increases by 1.55$\%$, but Round-Robin communication protocol leads
to about 50$\%$ bandwidth savings. This
example
shows that the distributed controllers subject
to Round-Robin communication protocol is more attractive when the network bandwidth is limited.
\end{example}

\begin{example}
This example considers a large-scale system in Fig.~\ref{figure1} with 100 heterogeneous subsystems. The results in \cite{ghadami2013decomposition} are not applicable because this large-scale system is not a ``decomposable system''.  The dynamics of the $i$th subsystem is described by equation \eqref{1}, where
\begin{align*}
A_{ii}&=\setlength{\arraycolsep}{0.5pt}\left\{
  \begin{array}{ll}
    \begin{bmatrix}
  \frac{2+{\rm mod}(i,3)}{10} & \frac{i}{5i+1}\\
   0 &  \frac{{\rm mod}(i,4)-15}{10}
  \end{bmatrix}, & \hbox{if $i=4m-1,m\in \mathbb{N}^{+}$};\\
   \begin{bmatrix}
   \frac{{\rm mod}(i,3)-20}{10} & \frac{i}{5i+1}\\
   0 &   \frac{{\rm mod}(i,4)-15}{10}
  \end{bmatrix}, & \hbox{otherwise},
  \end{array}
\right.\\
 A_{ij}&=\begin{bmatrix}   
 \frac{{\rm mod}(i,3)}{10} & \frac{{\rm mod}(i,4)}{10}\\
   0 & 0.1
  \end{bmatrix},~
  E_{i}=\begin{bmatrix}   
    \frac{{\rm mod}(i,6)}{10}\\
   0
  \end{bmatrix},\\
  B_{i}&=\begin{bmatrix}  
    -0.4\\
    0.1
  \end{bmatrix},~
  F_{i}= \frac{{\rm mod}(i,8)}{10},~
  C_{i}=\begin{bmatrix}
   \frac{{\rm mod}(i,5)}{10}  & 0.1
  \end{bmatrix}.
 \end{align*}

The compact form of the large-scale system is given by:
 \begin{align*}
 \dot{x}(t) = Ax(t)+Bu(t)+Ew(t),
 \end{align*}
 where $x(t)=[x_{1}^{\top}(t)\cdots x_{100}^{\top}(t)]^{\top}$, and $u(t),w(t)$ are defined accordingly,
\begin{align*}
A&=\begin{bmatrix}
   A_{1,1} & A_{1,2} & \mathbf{ 0}             & \cdots              & \mathbf{ 0}           &\mathbf{ 0} \\
   A_{2,1} & A_{2,2} & A_{2,3}         & \cdots              & \mathbf{ 0}           & \mathbf{ 0} \\
    \vdots & \vdots &  \vdots &        \ddots     & \vdots      & \vdots\\
  \mathbf{ 0}       & \mathbf{ 0}      & \mathbf{ 0}             & \cdots             &  A_{100,99}  & A_{100,100}
  \end{bmatrix},\\
  B&={\rm diag}(B_{1},\ldots,B_{100}),~
  E={\rm diag}(E_{1},\ldots, E_{100}).
\end{align*}

Note that for $i=4m-1$ $(m\in \mathbb{N}^{+})$, some eigenvalues of the matrix $A_{ii}$ locate  in the right half of the complex plane. Also, it can  be verified  that  12.5$\%$ eigenvalues of the matrix $A$ are in the right half of the complex plane. That is, the large-scale system is open-loop  unstable.

Under Round-Robin communication protocol, the suboptimal distributed controller \eqref{2} can be designed based on the optimization problem \eqref{1000}. The values of $\alpha_{i}$ $(i\in \mathbb{N})$, $h_{i}$  and $\Delta$ are chosen the same as those in Example \ref{example2}. By solving the optimization problem \eqref{1000} with 200 LMIs restrictions, we obtained that $\gamma_{\rm min}$=3.9255, all sub-controllers $K_{ii},K_{ij}$ ($i\in \mathbb{N},~j\in \mathbb{N}_{i}$) can be seen in the Appendix.  Set initial conditions as $x_{i}(0)=[1-2i,~2i]^{\top}$ $(i\in \mathbb{N})$, the state responses of the closed-loop system are shown in Fig.~\ref{figure99}, which has conformed that the controlled large-scale system is exponentially stable.
\begin{figure}[bht]
\centering
\includegraphics[width=0.43\textwidth]{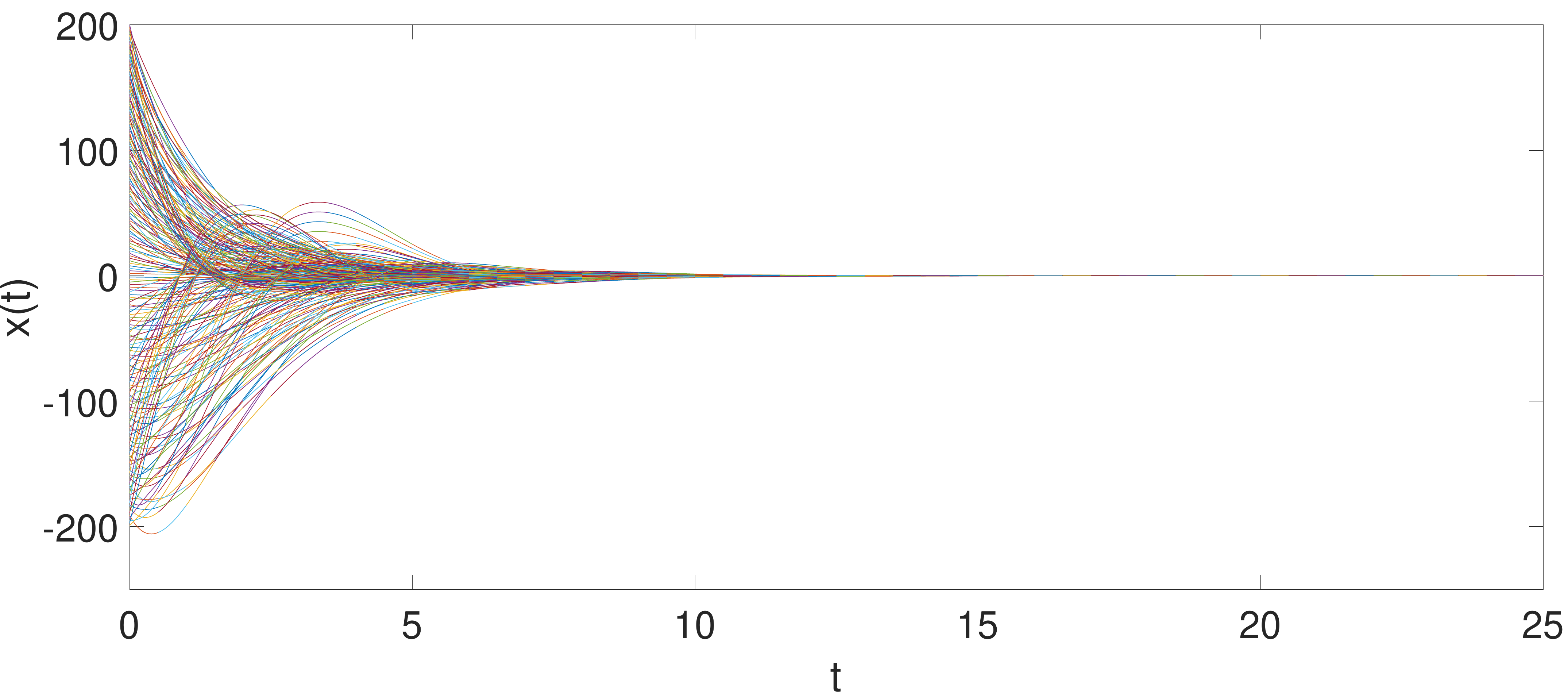}
\caption{State responses of the controlled large-scale system. }
\label{figure99}
\end{figure}
\end{example}

\section{Conclusions}
 This paper introduces a method to design distributed controllers for large-scale systems under  Round-Robin communication protocol.
 Based on a skillful partition of the time
interval $[0,T]$, a time-delay dependent approach has been introduced for the controller design and $\mathcal{L}_{2}$-gain analysis of the large-scale system. The distributed controller gains can be obtained by solving a set of LMIs. Finally, three numerical examples have demonstrated the validity of the proposed control scheme. Future work may involve the comparison of controller design and $\mathcal{L}_{2}$-gain of large-scale systems under different communication protocols, such as gossip communication protocol and try-once-discard communication protocol.



%

\section*{Appendix}
The distributed controllers obtained in Example 4 are given as follows:
  \newcounter{mytempeqncnt2}
\begin{figure*}[!bht]
\normalsize
\setcounter{mytempeqncnt2}{\value{equation}}
  \begin{align*}
[K_{1,1} \quad K_{1,2}]&=[-2.9881 \quad 0.4165\quad 	-1.4576\quad 	0.0325],\\
[K_{2,1} \quad K_{2,2}\quad K_{23}]&=[2.9873 \quad 	0.6547\quad 	-1.1356\quad 	0.556\quad 	0.5039\quad 	0.4013],\\
[K_{3,2} \quad K_{3,3}\quad K_{3,4}]&=[10.2471\quad	1.2361	\quad6.4269\quad	1.3454\quad	0.0023	\quad1.0177],\\
[K_{4,3} \quad K_{4,4}\quad K_{4,5}]&=[-6.8184\quad 0.232\quad	-3.538	\quad0.9847	\quad-0.2134	\quad0.2142	],\\
[K_{5,4} \quad K_{5,5}\quad K_{5,6}]&=[-10.2176\quad	-0.0347\quad	-7.1139\quad	0.8512\quad	-0.5442\quad	-0.0223],\\	
[K_{6,5} \quad K_{6,6}\quad K_{6,7}]&=[-2.4235\quad	-0.0609	\quad-2.3634\quad	-0.0137	\quad0.0001\quad	-0.0532],\\	
[K_{7,6} \quad K_{7,7}\quad K_{7,8}]&=[-26.8327\quad	-3.0094	\quad-28.9199\quad	-4.9757	\quad-1.1981\quad	-2.7055],\\
[K_{8,7} \quad K_{8,8}\quad K_{8,9}]&=[-2.6207	\quad0.0751\quad	-1.8963\quad	0.3957\quad	-0.1253\quad	0.0737],\\
[K_{9,8} \quad K_{9,9}\quad K_{9,10}]&=[-3.3166\quad	0.0389\quad	-3.681\quad	1.2635\quad	0.0001\quad	0.0657],\\
[K_{10,9} \quad K_{10,10}\quad K_{10,11}]&=[5.4765\quad	0.5221\quad	-0.7295\quad	0.556\quad	0.3024\quad	0.4221],\\
[K_{11,10} \quad K_{11,11}\quad K_{11,12}]&=[34.0197 \quad	3.9816\quad	22.5856\quad	5.269\quad	1.9759\quad	3.3009],\\
[K_{12,11} \quad K_{12,12}\quad K_{12,13}]&=[-1.5212\quad	0.0771\quad	-3.907\quad	0.2402\quad	0.0031\quad	0.0793],\\		
[K_{13,12} \quad K_{13,13}\quad K_{13,14}]&=[-17.3442\quad  -1.8496\quad	-9.7401\quad	3.8742\quad	-2.1945\quad	-0.8366],\\	
[K_{14,13} \quad K_{14,14}\quad K_{14,15}]&=[4.4353\quad	0.5791\quad	-1.6186\quad	0.6376\quad	0.4997\quad	0.4533],\\		
[K_{15,14} \quad K_{15,15}\quad K_{15,16}]&=[7.0064\quad	0.9356\quad	4.2748\quad	0.8026\quad	0.0016\quad	0.7724],\\			
[K_{16,15} \quad K_{16,16}\quad K_{16,17}]&=[-6.6066\quad	0.2202\quad	-3.4307\quad	0.8446\quad	-0.2084\quad	0.1914],\\	
[K_{17,16} \quad K_{17,17}\quad K_{17,18}]&=[-9.4721\quad	-0.0243\quad	-6.5963\quad	0.8454\quad	-0.4988\quad	-0.013],\\	
[K_{18,17} \quad K_{18,18}\quad K_{18,19}]&=[-2.4111\quad	-0.0608\quad	-2.3808\quad	-0.0183\quad	0.0001\quad	-0.0532],\\
[K_{19,18} \quad K_{19,19}\quad K_{19,20}]&=[-27.4712\quad	-3.0642\quad	-29.6985\quad	-5.1314\quad	-1.2608\quad	-2.7657],\\							
[K_{20,19} \quad K_{20,20}\quad K_{20,21}]&=[-2.4645\quad	0.0831\quad	-2.0535\quad	0.6284\quad	-0.1141\quad	0.0839],\\	
[K_{21,20} \quad  K_{21,21}\quad K_{21,22}]&=[-3.4489\quad 	-0.0132\quad 	-4.5794\quad 	1.5411\quad 	0\quad 	0.0547 ],\\		
[K_{22,21} \quad  K_{22,22}\quad K_{22,23}]&=[5.0341\quad 	0.5145\quad 	-1.0261\quad 	0.6615\quad 	0.2778\quad 	0.4248],\\		
[K_{23,22} \quad  K_{23,23}\quad K_{23,24}]&=[23.9335\quad 	2.8979\quad 	15.9821\quad 	3.4846\quad 	1.3469\quad 	2.3384],\\		
[K_{24,23} \quad  K_{24,24}\quad K_{24,25}]&=[-1.2408\quad 	0.0637\quad 	-4.3468\quad 	0.1721\quad 	0.0028\quad 	0.0645],\\		
[K_{25,24} \quad  K_{25,25}\quad K_{25,26}]&=[-19.3066\quad 	-3.1939\quad 	-20.4887\quad 	6.29\quad 	-3.7683\quad 	-1.583],\\	
[K_{26,25} \quad  K_{26,26}\quad K_{26,27}]&=[4.311\quad 	0.592\quad 	-1.6317\quad 	0.6286\quad 	0.4994\quad 	0.4548],\\		
[K_{27,26} \quad  K_{27,27}\quad K_{27,28}]&=[6.7997\quad 	0.9101\quad 	4.1681\quad 	0.7702\quad 	0.0016\quad 	0.7542],\\			
[K_{28,27} \quad  K_{28,28}\quad K_{28,29}]&=[-6.9866\quad 	0.2307\quad 	-3.6033\quad 	0.8703\quad 	-0.2239\quad 	0.2011],\\		
[K_{29,28} \quad  K_{29,29}\quad K_{29,30}]&=[-10.3709\quad 	-0.0359\quad 	-7.1703\quad 	0.797\quad 	-0.551\quad 	-0.023],\\		
[K_{30,29} \quad  K_{30,30}\quad K_{30,31}]&=[-2.4292\quad 	-0.0607\quad 	-2.3597\quad 	-0.0167\quad 	0.0001\quad 	-0.0532],\\	
[K_{31,30} \quad  K_{31,31}\quad K_{31,32}]&=[-26.9986\quad 	-3.0276\quad 	-29.0174\quad 	-5.054\quad 	-1.2043\quad 	-2.7202],\\	
[K_{32,31} \quad  K_{32,32}\quad K_{32,33}]&=[-2.6283\quad 	0.0751\quad 	-1.895\quad 	0.3903\quad 	-0.1258\quad 	0.0737],\\	
[K_{33,32} \quad  K_{33,33}\quad K_{33,34}]&=[-3.3259\quad 	0.0397\quad 	-3.6654\quad 	1.2573\quad 	0.0001\quad 	0.0659],\\	
[K_{34,33} \quad  K_{34,34}\quad K_{34,35}]&=[5.5037\quad 	0.5223\quad 	-0.7044\quad 	0.5604\quad 	0.3032\quad 	0.421],\\		
[K_{35,34} \quad  K_{35,35}\quad K_{35,36}]&=[37.8394\quad 	4.4204\quad 	25.0563\quad 	5.8795\quad 	1.7876\quad 	3.4566],\\	
[K_{36,35} \quad  K_{36,36}\quad K_{36,37}]&=[-0.7023\quad 	0.0285\quad 	-5.1411\quad 	0.0218\quad 	0.0017\quad 	0.029],\\		
[K_{37,36} \quad  K_{37,37}\quad K_{37,38}]&=[-21.9204\quad 	2.5976\quad 	45.8645\quad 	-5.9935\quad 	3.4251\quad 	1.9281],\\	
[K_{38,37} \quad  K_{38,38}\quad K_{38,39}]&=[3.7914\quad 	0.6048\quad 	-1.4015\quad 	0.6511\quad 	0.4988\quad 	0.4546],\\		
[K_{39,38} \quad  K_{39,39}\quad K_{39,40}]&=[6.9622\quad 	0.9241\quad 	4.3634\quad 	0.8188\quad 	0.0017\quad 	0.761],\\			
[K_{40,39} \quad  K_{40,40}\quad K_{40,41}]&=[-6.6236\quad 	0.2207\quad 	-3.4382\quad 	0.8357\quad 	-0.2093\quad 	0.1908],\\
[K_{41,40} \quad  K_{41,41}\quad K_{41,42}]&=[-9.4861\quad 	-0.0245\quad 	-6.6092\quad 	0.838\quad 	-0.5\quad 	-0.0133],\\
[K_{42,41} \quad  K_{42,42}\quad K_{42,43}]&=[-2.417\quad 	-0.0611\quad 	-2.3762\quad 	-0.0193\quad 	0.0001\quad 	-0.0534],\\		
[K_{43,42} \quad  K_{43,43}\quad K_{43,44}]&=[-26.51\quad 	-2.9523\quad 	-28.5949\quad 	-4.9447\quad 	-1.1916\quad 	-2.6658],\\	
[K_{44,43} \quad  K_{44,44}\quad K_{44,45}]&=[-2.6286\quad 	0.078\quad 	-1.9348\quad 	0.4502\quad 	-0.1246\quad 	0.0774],\\	
[K_{45,44} \quad  K_{45,45}\quad K_{45,46}]&=[-3.4629\quad 	0.0265\quad 	-4.0645\quad 	1.4109\quad 	0\quad 	0.0616],\\	
[K_{46,45} \quad  K_{46,46}\quad K_{46,47}]&=[5.3378\quad 	0.5212\quad 	-0.81\quad 	0.6161\quad 	0.2932\quad 	0.4284],
\end{align*}
\end{figure*}

\begin{figure*}[!bht]
\normalsize
\setcounter{mytempeqncnt2}{\value{equation}}
  \begin{align*}			
[K_{47,46} \quad  K_{47,47}\quad K_{47,48}]&=[29.7628\quad 	3.5448\quad 	19.8914\quad 	4.6098\quad 	1.7563\quad 	2.9094],\\		
[K_{48,47} \quad  K_{48,48}\quad K_{48,49}]&=[-1.5331\quad 	0.0805\quad 	-3.875\quad 	0.2498\quad 	0.0034\quad 	0.0816],\\			
[K_{49,48} \quad  K_{49,49}\quad K_{49,50}]&=[-12.8387\quad 	-1.3657\quad 	-7.5098\quad 	2.931\quad 	-1.5947\quad 	-0.5824],\\	
[K_{50,49} \quad  K_{50,50}\quad K_{50,51}]&=[4.3816\quad 	0.5849\quad 	-1.6544\quad 	0.6344\quad 	0.4996\quad 	0.4521],\\		
[K_{51,50} \quad  K_{51,51}\quad K_{51,52}]&=[6.9543\quad 	0.9238\quad 	4.2539\quad 	0.7912\quad 	0.0016\quad 	0.7624],\\			
[K_{52,51} \quad  K_{52,52}\quad K_{52,53}]&=[-7.0183\quad 	0.2389\quad 	-3.6328\quad 	0.9049\quad 	-0.2234\quad 	0.1991],\\	
[K_{53,52} \quad  K_{53,53}\quad K_{53,54}]&=[-11.4582\quad 	0.0975\quad 	-8.068\quad 	2.2294\quad 	-0.4915\quad 	0.0795],\\		
[K_{54,53} \quad  K_{54,54}\quad K_{54,55}]&=[-1.3728\quad 	0.0271\quad 	-3.1389\quad 	0.1263\quad 	-0.0014\quad 	0.0229],\\	
[K_{55,54} \quad  K_{55,55}\quad K_{55,56}]&=[12.5213\quad 	4.3966\quad 	29.5485\quad 	14.4651\quad 	1.5009\quad 	2.2636],\\	
[K_{56,55} \quad  K_{56,56}\quad K_{56,57}]&=[-1.8283\quad 	0.1593\quad 	-2.4882\quad 	1.4566\quad 	-0.044\quad 	0.1602],\\		
[K_{57,56} \quad  K_{57,57}\quad K_{57,58}]&=[-7.0562\quad 	0.0127\quad 	-12.9813\quad 	6.0704\quad 	-0.0003\quad 	0.5065],\\	
[K_{58,57} \quad  K_{58,58}\quad K_{58,59}]&=[4.7371\quad 	0.4145\quad 	-1.2642\quad 	0.6963\quad 	0.2544\quad 	0.3591],\\		
[K_{59,58} \quad  K_{59,59}\quad K_{59,60}]&=[13.4463\quad 	1.5894\quad 	8.7723\quad 	2.6072\quad 	0.5721\quad 	0.947],\\	
[K_{60,59} \quad  K_{60,60}\quad K_{60,61}]&=[-0.852\quad 	0.041\quad 	-5.0231\quad 	0.0092\quad 	0.0022\quad 	0.0441],\\	
[K_{61,60} \quad  K_{61,61}\quad K_{61,62}]&=[-5.3235\quad 	2.1274\quad 	53.2428\quad 	-10.6633\quad 	6.388\quad 	3.1399],\\	
[K_{62,61} \quad  K_{62,62}\quad K_{62,63}]&=[4.1144\quad 	0.6045\quad 	-1.533\quad 	0.6397\quad 	0.4993\quad 	0.457],\\		
[K_{63,62} \quad  K_{63,63}\quad K_{63,64}]&=[6.7766\quad 	0.9132\quad 	4.2089\quad 	0.789\quad 	0.0016\quad 	0.7524],\\			
[K_{64,63} \quad  K_{64,64}\quad K_{64,65}]&=[-6.628\quad 	0.2192\quad 	-3.4373\quad 	0.8245\quad 	-0.2098\quad 	0.1891],\\		
[K_{65,64} \quad  K_{65,65}\quad K_{65,66}]&=[-9.494\quad 	-0.025\quad 	-6.6168\quad 	0.8389\quad 	-0.5009\quad 	-0.0137],\\		
[K_{66,65} \quad  K_{66,66}\quad K_{66,67}]&=[-2.4209\quad 	-0.0613\quad 	-2.3731\quad 	-0.0196\quad 	0.0001\quad 	-0.0536],\\																						 [K_{67,66} \quad  K_{67,67}\quad K_{67,68}]&=[-26.3829\quad 	-2.9379\quad 	-28.411\quad 	-4.9161\quad 	-1.1808\quad 	-2.6527],\\	
[K_{68,67} \quad  K_{68,68}\quad K_{68,69}]&=[-2.6794\quad 	0.0773\quad 	-1.909\quad 	0.4112\quad 	-0.1287\quad 	0.0764],\\	
[K_{69,68} \quad  K_{69,69}\quad K_{69,70}]&=[-3.4884\quad 	0.0358\quad 	-3.8506\quad 	1.354\quad 	0.0001\quad 	0.0644],\\	
[K_{70,69} \quad  K_{70,70}\quad K_{70,71}]&=[5.5146\quad 	0.5234\quad 	-0.7174\quad 	0.5887\quad 	0.3014\quad 	0.431],\\		
[K_{71,70} \quad  K_{71,71}\quad K_{71,72}]&=[32.5948\quad 	3.8555\quad 	21.7134\quad 	5.1788\quad 	1.9238\quad 	3.1833],\\		
[K_{72,71} \quad  K_{72,72}\quad K_{72,73}]&=[-1.5517\quad 	0.0809\quad 	-3.8502\quad 	0.2529\quad 	0.0034\quad 	0.0822],\\		
[K_{73,72} \quad  K_{73,73}\quad K_{73,74}]&=[-13.5249\quad 	-1.4\quad 	-7.3738\quad 	2.9665\quad 	-1.6447\quad 	-0.5954],\\	
[K_{74,73} \quad  K_{74,74}\quad K_{74,75}]&=[4.4017\quad 	0.5811\quad 	-1.6477\quad 	0.6365\quad 	0.4996\quad 	0.4515],\\		
[K_{75,74} \quad  K_{75,75}\quad K_{75,76}]&=[6.9884\quad 	0.9269\quad 	4.2656\quad 	0.7956\quad 	0.0016\quad 	0.768],\\			
[K_{76,75} \quad  K_{76,76}\quad K_{76,77}]&=[-6.9674\quad 	0.2317\quad 	-3.6034\quad 	0.8836\quad 	-0.223\quad 	0.2031],\\			
[K_{77,76} \quad  K_{77,77}\quad K_{77,78}]&=[-10.3523\quad 	-0.0357\quad 	-7.1903\quad 	0.8117\quad 	-0.5506\quad 	-0.0227],\\	
[K_{78,77} \quad  K_{78,78}\quad K_{78,79}]&=[-2.4165\quad 	-0.0601\quad 	-2.3709\quad 	-0.0168\quad 	0.0001\quad 	-0.0526],\\	
[K_{79,78} \quad  K_{79,79}\quad K_{79,80}]&=[-27.61\quad 	-3.0968\quad 	-29.9205\quad 	-5.2211\quad 	-1.2532\quad 	-2.7825],\\	
[K_{80,79} \quad  K_{80,80}\quad K_{80,81}]&=[-2.4788\quad 	0.0791\quad 	-1.9923\quad 	0.5373\quad 	-0.115\quad 	0.0789],\\
[K_{81,80} \quad  K_{81,81}\quad K_{81,82}]&=[-3.2696\quad 	0.0092\quad 	-4.1932\quad 	1.3945\quad 	0\quad 	0.0606],\\	 	
[K_{82,81} \quad  K_{82,82}\quad K_{82,83}]&=[5.0672\quad 	0.515\quad 	-0.9945\quad 	0.6344\quad 	0.2825\quad 	0.4175],\\	
[K_{83,82} \quad  K_{83,83}\quad K_{83,84}]&=[26.3286\quad 	3.141\quad 	17.4066\quad 	3.8426\quad 	1.4524\quad 	2.5498],\\
[K_{84,83} \quad  K_{84,84}\quad K_{84,85}]&=[-1.1965\quad 	0.0595\quad 	-4.4272\quad 	0.156\quad 	0.0027\quad 	0.0611],\\		
[K_{85,84} \quad  K_{85,85}\quad K_{85,86}]&=[-35.7404\quad 	-6.6442\quad 	-44.7138\quad 	12.869\quad 	-7.9382\quad 	-3.4925],\\	
[K_{86,85} \quad  K_{86,86}\quad K_{86,87}]&=[4.3246\quad 	0.5914\quad 	-1.6145\quad 	0.6381\quad 	0.4995\quad 	0.4564],\\		
[K_{87,86} \quad  K_{87,87}\quad K_{87,88}]&=[6.8172\quad 	0.9182\quad 	4.2042\quad 	0.7901\quad 	0.0016\quad 	0.7574],\\				
[K_{88,87} \quad  K_{88,88}\quad K_{88,89}]&=[-6.6285\quad 	0.2194\quad 	-3.4382\quad 	0.8288\quad 	-0.2098\quad 	0.1897],\\	
[K_{89,88} \quad  K_{89,89}\quad K_{89,90}]&=[-9.4979\quad 	-0.0249\quad 	-6.6046\quad 	0.8321\quad 	-0.5005\quad 	-0.0135],\\	
[K_{90,89} \quad  K_{90,90}\quad K_{90,91}]&=[-2.423\quad 	-0.0613\quad 	-2.3712\quad 	-0.0198\quad 	0.0001\quad 	-0.0537],\\	
[K_{91,90} \quad  K_{91,91}\quad K_{91,92}]&=[-26.3586\quad 	-2.9351\quad 	-28.3509\quad 	-4.9078\quad 	-1.1789\quad 	-2.6498],\\													[K_{92,91} \quad  K_{92,92}\quad K_{92,93}]&=[-2.6833\quad 	0.0773\quad 	-1.9083\quad 	0.4083\quad 	-0.1292\quad 	0.0764],
\end{align*}
\end{figure*}

\begin{figure*}[t]
\normalsize
\setcounter{mytempeqncnt2}{\value{equation}}
  \begin{align*}	
[K_{93,92} \quad  K_{93,93}\quad K_{93,94}]&=[-3.4982\quad 	0.0357\quad 	-3.8512\quad 	1.3488\quad 	0.0001\quad 	0.0638],\\	
[K_{94,93} \quad  K_{94,94}\quad K_{94,95}]&=[5.5228\quad 	0.5227\quad 	-0.6696\quad 	0.5871\quad 	0.3009\quad 	0.4298],\\		
[K_{95,94} \quad  K_{95,95}\quad K_{95,96}]&=[34.8414\quad 	4.1204\quad 	23.0285\quad 	5.49\quad 	1.6694\quad 	3.2287],\\	
[K_{96,95} \quad  K_{96,96}\quad K_{96,97}]&=[-0.7316\quad 	0.0305\quad 	-5.1084\quad 	0.0292\quad 	0.0018\quad 	0.0308],\\	
[K_{97,96} \quad  K_{97,97}\quad K_{97,98}]&=[-26.5513\quad 	3.803\quad 	64.5189\quad 	-9.1848\quad 	5.1759\quad 	2.7955],\\	
[K_{98,97} \quad  K_{98,98}\quad K_{98,99}]&=[3.8297\quad 	0.6094\quad 	-1.4615\quad 	0.6398\quad 	0.4987\quad 	0.4543],\\	
[K_{99,98} \quad  K_{99,99}\quad K_{99,100}]&=[6.8023\quad 	0.9071\quad 	4.2065\quad 	0.7711\quad 	0.0006\quad 	0.7979],\\	
[K_{100,99} \quad  K_{100,100} ]&=[-5.6206\quad 	0.1551\quad 	-3.1762\quad 	0.4535 ].\\
~\\~\\~\\~\\~\\~\\~\\~\\~\\~\\~\\~\\~\\~\\~\\~\\~\\~\\~\\~\\~\\~\\~\\~\\~\\~\\~\\~\\~\\~\\~\\~\\~\\~\\~\\~\\~\\~\\~\\~\\
~\\~\\~\\~\\~\\~\\~\\~\\~\\~\\~\\~\\~\\~\\~\\~\\
\end{align*}
\end{figure*}

\begin{thebibliography}{10}
\providecommand{\url}[1]{#1}
\csname url@samestyle\endcsname
\providecommand{\newblock}{\relax}
\providecommand{\bibinfo}[2]{#2}
\providecommand{\BIBentrySTDinterwordspacing}{\spaceskip=0pt\relax}
\providecommand{\BIBentryALTinterwordstretchfactor}{4}
\providecommand{\BIBentryALTinterwordspacing}{\spaceskip=\fontdimen2\font plus
\BIBentryALTinterwordstretchfactor\fontdimen3\font minus
  \fontdimen4\font\relax}
\providecommand{\BIBforeignlanguage}[2]{{%
\expandafter\ifx\csname l@#1\endcsname\relax
\typeout{** WARNING: IEEEtran.bst: No hyphenation pattern has been}%
\typeout{** loaded for the language `#1'. Using the pattern for}%
\typeout{** the default language instead.}%
\else
\language=\csname l@#1\endcsname
\fi
#2}}
\providecommand{\BIBdecl}{\relax}
\BIBdecl

\bibitem{kazempour2013stability}
F.~Kazempour and J.~Ghaisari, ``Stability analysis of model-based networked
  distributed control systems,'' \emph{Journal of Process Control}, vol.~23,
  no.~3, pp. 444--452, 2013.

\bibitem{massioni2009distributed}
P.~Massioni and M.~Verhaegen, ``Distributed control for identical dynamically
  coupled systems: A decomposition approach,'' \emph{IEEE Transactions on
  Automatic Control}, vol.~54, no.~1, pp. 124--135, 2009.

\bibitem{sun2018distributed}
J.~Sun, Z.~Geng, Y.~Lv, Z.~Li, and Z.~Ding, ``Distributed adaptive consensus
  disturbance rejection for multi-agent systems on directed graphs,''
  \emph{IEEE Transactions on Control of Network Systems}, vol.~5, no.~1, pp.
  629--639, 2018.

\bibitem{boem2019distributed}
F.~Boem, A.~Gallo, D.~M. Raimondo, and T.~Parisini, ``Distributed
  fault-tolerant control of large-scale systems: an active fault diagnosis
  approach,'' \emph{IEEE Transactions on Control of Network Systems}, 2019.
  doi: 10.1109/TCNS.2019.2913557

\bibitem{zhang2016energy}
D.~Zhang, P.~Shi, and Q.-G. Wang, ``Energy-efficient distributed control of
  large-scale systems: A switched system approach,'' \emph{International
  Journal of Robust and Nonlinear Control}, vol.~26, no.~14, pp. 3101--3117,
  2016.

\bibitem{millan2019distributed}
P.~Mill{\'a}n, L.~Orihuela, and I.~Jurado, ``Distributed agent-based control
  and estimation over unreliable networks for a class of nonlinear large-scale
  systems,'' \emph{International Journal of Control}, vol.~92, no.~3, pp.
  664--676, 2019.

\bibitem{van2015synthesis}
E.~P. van Horssen and S.~Weiland, ``Synthesis of distributed robust
  ${H_{\infty}}$ controllers for interconnected discrete time systems,''
  \emph{IEEE Transactions on Control of Network Systems}, vol.~3, no.~3, pp.
  286--295, 2015.

\bibitem{ugrinovskii2014Round}
V.~Ugrinovskii and E.~Fridman, ``A {Round-Robin} type protocol for distributed
  estimation with ${H_{\infty}}$ consensus,'' \emph{Systems $\&$ Control
  Letters}, vol.~69, pp. 103--110, 2014.

\bibitem{zou2017state}
L.~Zou, Z.~Wang, H.~Gao, and X.~Liu, ``State estimation for discrete-time
  dynamical networks with time-varying delays and stochastic disturbances under
  the {Round-Robin} protocol,'' \emph{IEEE Transactions on Neural Networks and
  Learning Systems}, vol.~28, no.~5, pp. 1139--1151, 2017.

\bibitem{zou2017ultimate}
L.~Zou, Z.~Wang, Q.-L. Han, and D.~Zhou, ``Ultimate boundedness control for
  networked systems with try-once-discard protocol and uniform quantization
  effects,'' \emph{IEEE Transactions on Automatic Control}, vol.~62, no.~12,
  pp. 6582--6588, 2017.

\bibitem{donkers2012stability}
M.~Donkers, W.~Heemels, D.~Bernardini, A.~Bemporad, and V.~Shneer, ``Stability
  analysis of stochastic networked control systems,'' \emph{Automatica},
  vol.~48, no.~5, pp. 917--925, 2012.

\bibitem{yutao}
T.~Yu and J.~Xiong, ``Distributed ${L}_{2}$-gain control of large-scale systems
  under gossip communication protocol,'' \emph{International Journal of
  Control}, 2019. doi: 10.1080/00207179.2019.1631489

\bibitem{nesic2004input}
D.~Nesic and A.~R. Teel, ``Input-output stability properties of networked
  control systems,'' \emph{IEEE Transactions on Automatic Control}, vol.~49,
  no.~10, pp. 1650--1667, 2004.

\bibitem{heemels2010networked}
W.~M.~H. Heemels, A.~R. Teel, N.~Van~de Wouw, and D.~Nesic, ``Networked control
  systems with communication constraints: Tradeoffs between transmission
  intervals, delays and performance,'' \emph{IEEE Transactions on Automatic
  control}, vol.~55, no.~8, pp. 1781--1796, 2010.

\bibitem{xu2018finite}
Y.~Xu, R.~Lu, P.~Shi, H.~Li, and S.~Xie, ``Finite-time distributed state
  estimation over sensor networks with {Round-Robin} protocol and fading
  channels,'' \emph{IEEE Transactions on Cybernetics}, vol.~48, no.~1, pp.
  336--345, 2018.

\bibitem{liu2012stability}
K.~Liu, E.~Fridman, and L.~Hetel, ``Stability and $\mathcal{L}_{2}$-gain
  analysis of networked control systems under {Round-Robin} scheduling: A
  time-delay approach,'' \emph{Systems \& Control Letters}, vol.~61, no.~5, pp.
  666--675, 2012.

\bibitem{ghadami2013decomposition}
R.~Ghadami and B.~Shafai, ``Decomposition-based distributed control for
  continuous-time multi-agent systems,'' \emph{IEEE Transactions on Automatic
  Control}, vol.~58, no.~1, pp. 258--264, 2013.

\bibitem{chen2016distributed}
J.~Chen, R.~Ling, and D.~Zhang, ``Distributed non-fragile stabilization of
  large-scale systems with random controller failure,'' \emph{Neurocomputing},
  vol. 173, pp. 2033--2038, 2016.

\bibitem{wu2014distributed}
H.-N. Wu and H.-D. Wang, ``Distributed consensus observers-based ${H_{\infty}}$
  control of dissipative {PDE} systems using sensor networks,'' \emph{IEEE
  Transactions on Control of Network Systems}, vol.~2, no.~2, pp. 112--121,
  2014.

\bibitem{lee2006sufficient}
K.~H. Lee, J.~H. Lee, and W.~H. Kwon, ``Sufficient {LMI} conditions for
  ${H_{\infty}}$ output feedback stabilization of linear discrete-time
  systems,'' \emph{IEEE Transactions on Automatic Control}, vol.~51, no.~4, pp.
  675--680, 2006.

\bibitem{wang2007robust}
Z.~Wang, F.~Yang, D.~W. Ho, and X.~Liu, ``Robust ${H_{\infty}}$ control for
  networked systems with random packet losses,'' \emph{IEEE Transactions on
  Systems, Man, and Cybernetics---Part B}, vol.~37, no.~4, pp. 916--924, 2007.

\bibitem{liu2012wirtinger}
K.~Liu and E.~Fridman, ``Wirtinger’s inequality and {L}yapunov-based
  sampled-data stabilization,'' \emph{Automatica}, vol.~48, no.~1, pp.
  102--108, 2012.

\bibitem{seuret2015stability}
A.~Seuret, F.~Gouaisbaut, and E.~Fridman, ``Stability of discrete-time systems
  with time-varying delays via a novel summation inequality,'' \emph{IEEE
  Transactions on Automatic Control}, vol.~60, no.~10, pp. 2740--2745, 2015.

\bibitem{massioni2014distributed}
P.~Massioni, ``Distributed control for alpha-heterogeneous dynamically coupled
  systems,'' \emph{Systems $\&$ Control Letters}, vol.~72, pp. 30--35, 2014.

\bibitem{hoffmann2013distributed}
C.~Hoffmann, A.~Eichler, and H.~Werner, ``Distributed control of linear
  parameter-varying decomposable systems,'' in \emph{Proceedings of American
  Control Conference}, 2013, pp. 2380--2385.

\bibitem{wu2017cooperative}
J.~Wu, V.~Ugrinovskii, and F.~Allg{\"o}wer, ``Cooperative estimation and robust
  synchronization of heterogeneous multi-agent systems with coupled
  measurements,'' \emph{IEEE Transactions on Control of Network Systems},
  vol.~5, no.~4, pp. 1597--1607, 2018.

\end{thebibliography}
\end{document}